\documentclass[12pt,oneside,reqno]{article}
\makeatletter
\newcommand{\singlespacing}{\let\CS=\@currsize\renewcommand{\baselinestretch}{1}\tiny\CS}
 \usepackage{a4wide}
 \usepackage[a4paper, total={6.4in, 9.5in}]{geometry}
\usepackage{epsfig}
\usepackage{lineno}
\usepackage{rotating}
\usepackage{booktabs}
\usepackage{wrapfig}
\usepackage{comment}
\usepackage{multirow}
\usepackage{siunitx}
\usepackage{rotating}
\usepackage[T1]{fontenc}
\usepackage[utf8]{inputenc}
\usepackage{amsmath}
\usepackage{amsthm}
\newtheorem{lemma}{Lemma}
\usepackage{pdfpages}
\usepackage{amsfonts,graphicx}
\usepackage{authblk}
\usepackage{verbatim}
\usepackage[mathscr]{eucal}
\usepackage{booktabs}
\usepackage{siunitx}

\usepackage{mathptmx} %for curly LL sign
\usepackage{adjustbox}
\usepackage{xcolor}
\usepackage{xspace}
\usepackage{enumitem}
\usepackage{subfig}
\usepackage[colorlinks=true,urlcolor=blue,citecolor=blue,linkcolor=blue,bookmarks=true]{hyperref}

\usepackage{cleveref}
\crefname{appendix}{Appendix}{Appendices}
\usepackage{natbib}
\setcitestyle{authoryear,open={},close={}} 
\usepackage{float}
\usepackage{epstopdf}
\usepackage[parfill]{parskip}
\newcommand{\RN}[1]{%
	\textup{\uppercase\expandafter{\romannumeral#1}}%
}

\newcommand{\be}{\begin{equation}}
\newcommand{\ee}{\end{equation}}
\newcommand{\beanno}{\begin{eqnarray*}}

\newcommand{\eeanno}{\end{eqnarray*}}
\newcommand{\bea}{\begin{eqnarray}}
\newcommand{\eea}{\end{eqnarray}}
\newcommand{\ba}{\begin{array}}
\newcommand{\ea}{\end{array}}

\newcommand{\bc}{\begin{center}}
\newcommand{\ec}{\end{center}}

\newcommand{\ncom}{\newcommand}

\ncom{\half}{\frac{1}{2}}
\ncom{\mbx}{\makebox{.25cm}}
\ncom{\hs}{\mbox{\hspace{.25cm}}}
\ncom{\rar}{\rightarrow}
\ncom{\Rar}{\Rightarrow}
\ncom{\noin}{\noindent}
\ncom{\sz}{\scriptsize}
\ncom{\rf}{\ref}
\ncom{\s}{\sqrt{2}}
\ncom{\sgm}{\sigma}
\ncom{\Sgm}{\Sigma}
\ncom{\psgm}{\sigma^{\prime}}
\ncom{\dt}{\delta}
\ncom{\Dt}{\Delta}
\ncom{\lmd}{\la}
\ncom{\Lmd}{\la}
\ncom{\Th}{\Theta}
\ncom{\e}{\xi}
\ncom{\eps}{\epsilon}
\ncom{\pcc}{\stackrel{P}{>}}
\ncom{\lp}{\stackrel{L_{p}}{>}}
\ncom{\dist}{{\rm\,dist}}
\ncom{\sspan}{{\rm\,span}}
\ncom{\re}{{\rm Re\,}}
\ncom{\im}{{\rm Im\,}}
\ncom{\sgn}{{\rm sgn\,}}
\clearpage
\ncom{\hone}{\mbox{\hspace{1em}}}
\ncom{\htwo}{\mbox{\hspace{2em}}}
\ncom{\hthree}{\mbox{\hspace{3em}}}
\ncom{\hfour}{\mbox{\hspace{4em}}}
\ncom{\vone}{\vskip 2ex}
\ncom{\vtwo}{\vskip 4ex}
\ncom{\vonee}{\vskip 1.5ex}
\ncom{\vthree}{\vskip 6ex}
\ncom{\vfour}{\vspace*{8ex}}
\ncom{\norm}{\|\;\;\|}
\ncom{\integ}[4]{\int_{#1}^{#2}\,{#3}\,d{#4}}
\ncom{\vspan}[1]{{{\rm\,span}\{ #1 \}}}
\ncom{\dm}[1]{ {\displaystyle{#1} } }
\ncom{\ri}[1]{{#1} \index{#1}}

\newtheoremstyle
{remarkstyle}
{}
{11pt}
{}
{}
{\bfseries}
{:}
{     }
{\thmname{#1} \thmnumber{#2} }
\theoremstyle{remarkstyle}

\renewcommand{\baselinestretch}{1.5}
\def \x {x}

\def \y {y}

\def \a {a}
\def \b {b}

\def \f {f}
\def \F {F}

\def \fx {\psi(\x; \lambda)}
\def \fm {\psi(\mu_{\rho};\lambda)}
\def \fmt {\psi(\mu_{\rho,t};\lambda)}

\def \fy {\psi(\y;\lambda)}
\def \fyt {\psi(\y_t;\lambda)}

\allowdisplaybreaks
\begin{document}
	\title{\bf A Quantile-Based Kumaraswamy–Teissier autoregressive moving average models}
\author{ Kamana Mishra$^1$, Tanmay Kayal$^1$ and Sarita Azad$^1$ \thanks{Corresponding Author Email: sarita@iitmandi.ac.in}}
\affil{\it $^1$School of Mathematical and Statistical Sciences \\ \it Indian Institute of Technology Mandi, Himachal Pradesh-175005, India }
	\date{}
	\maketitle
	\vspace{-1.5em}

\begin{abstract}
\noindent

This paper introduces a quantile-based Kumaraswamy–Teissier autoregressive moving average (KTARMA) model for positive-valued time series. Leveraging the flexibility of the extended Kumaraswamy–Teissier distribution within an observation-driven framework, the random component of the distribution is conditioned on the historical process and time-varying covariates, and is parameterized explicitly via its $\rho$-th conditional quantile, where $\rho \in (0,1)$. To capture temporal dependence, the systematic component maps an ARMA-type structure to this conditional quantile via an appropriate link function. For inference, we implement a conditional maximum likelihood framework and derive explicit analytical expressions for the resulting score vector and conditional information matrix, followed by the development of model diagnostic and forecasting procedures. The finite-sample performance of the developed estimators is evaluated through a Monte Carlo simulation study across various parameter configurations and quantile levels. Finally, the practical utility of the study is demonstrated by modeling monthly rainfall data over the Northwest Himalayas (2001–2025), where 525 grids are grouped into four homogeneous zones using a Self-Organizing Map and 
relevant atmospheric variables and large-scale climate indices are incorporated as predictive regressors. Out-of-sample forecasting evaluations reveal that the KTARMA model delivers highly competitive predictive performance, achieving consistently lower mean squared errors across all identified zones compared to KARMA and $\beta$ARMA models.
\end{abstract}


\section{Introduction}
Time series often exhibit seasonality, asymmetry, and serial dependence. The classical autoregressive integrated moving average (ARIMA) models provide an important framework for analysing temporal dependence (\cite{box2015time}). However, its underlying assumption of normality is poorly suited for highly skewed non-Gaussian environmental variables. A major limitation of this approach is its tendency to yield out-of-sample forecasts that violate the natural bounded support of the observed process (\cite{tiku2000time}). These limitations have motivated the development of distribution-based time-series models that explicitly account for the marginal characteristics of the observations while simultaneously describing their temporal dependence. An important development in this direction is the generalized autoregressive moving average (GARMA) framework of (\cite{benjamin2003generalized}), which extends the ARMA structure to non-Gaussian responses by linking a conditional distribution parameter to a dynamic systematic component. For processes taking values in the unit interval, Rocha and Cribari-Neto introduced the $\beta$ARMA model by combining the beta distribution with an ARMA-type structure (\cite{rocha2009beta}). The model was subsequently extended in the literature, including further developments for fractionally integrated dynamics (\cite{rocha2017erratum,pumi2019beta}). The Kumaraswamy distribution has also emerged as a useful alternative because of its flexibility in representing a broad range of distributional shapes (\cite{ nadarajah2008distribution,lemonte2013exponentiated}). Its use in regression modelling has also been studied, including formulations with alternative link functions such as the Aranda--Ordaz link (\cite{pumi2020kumaraswamy}). Building on this distributional framework, \cite{bayer2017kumaraswamy} proposed the Kumaraswamy autoregressive moving average (KARMA) model, in which the conditional median is dynamically related to covariates through autoregressive and moving-average terms and a suitable link function. The use of the conditional median is particularly appealing for skewed time series, as it is less sensitive to extreme observations and outliers than the mean (\cite{john2015robustness}). 

While the KARMA framework provides a useful dynamic model for bounded environmental processes, its formulation is centered on the conditional median. This restriction may be insufficient when interest extends beyond the central part of the conditional distribution, particularly in applications where lower or upper rainfall regimes are of practical importance. A quantile-based formulation provides a more flexible alternative by allowing different parts of the conditional distribution to be modelled through a specified quantile level $\rho\in(0,1)$. In this context, we employ the Kumaraswamy--Teissier distribution (KTD), introduced by (\cite{mishra2026}), as the underlying distributional framework. The KTD is formulated for positive-valued observations and provides a flexible marginal distribution for rainfall-type data. However, as a static distribution, it does not account for the serial dependence commonly present in rainfall time series. We therefore combine the KTD with an ARMA-type dynamic structure to develop a time-series model for rainfall. Furthermore, following the quantile-based formulation of the Unit-Weibull ARMA model proposed by (\cite{pumi2024unit}), which directly models the conditional $\rho$th quantile for $\rho\in(0,1)$, the proposed KTARMA model extends the median-based framework to a general conditional quantile formulation.

The need for such a flexible framework is particularly relevant for rainfall over the Northwest Himalayas (NWH), a region characterized by strong spatial and temporal heterogeneity. Recent studies report substantial rainfall variability and evolving spatial patterns, with long-term changes linked to atmospheric circulation, temperature, elevation, and large-scale climatic forcing (\cite{sharma2025spatio,banerjee2023solid,jena2019weakening,mishra2026spatiotemporal,singh1997effect}). Large-scale teleconnections also contribute to regional rainfall variability, with their influence varying across seasons and sub-regions (\cite{bhutiyani2010climate,mal2021spatial}). This motivates statistical models that can accommodate serial dependence, distributional asymmetry, seasonal variation, and relevant climatic covariates.

Building on these methodological and application considerations, the main contributions of this work are threefold. First, we introduce a new observation-driven KTARMA model that combines the Kumaraswamy--Teissier distribution with an ARMA structure for positive-valued time series. Second, we extend the conventional median-based dynamic modelling framework to a general conditional quantile formulation, allowing different regions of the conditional rainfall distribution to be examined through $\rho\in(0,1)$. Third, we demonstrate the practical utility of the proposed framework through an application to spatially heterogeneous NWH rainfall and evaluate its out-of-sample forecasting performance against competing models.

The rest of this paper is organized as follows. Section \ref{ktdarma} presents the formulation of the proposed KTARMA model. Parameter estimation via the conditional maximum likelihood method, together with the derivation of the corresponding score vector and conditional information matrix, is developed in Section \ref{cmle}. Model diagnostic and forecasting procedures are presented in Section \ref{diagnostic}. Section \ref{simulation} investigates the finite-sample properties of the proposed estimators through Monte Carlo simulations. Section \ref{application} presents the application to rainfall modelling over the NWH and compares the forecasting performance of the proposed model with competing models. Section \ref{conclusion} concludes the paper, while technical results and proofs are provided in the Appendix.

\section{Kumaraswamy Teissier Auto-regressive moving average model (KTARMA)} \label{ktdarma}
We have proposed the KTARMA model by integrating the dynamic ARMA structure with the Kumaraswamy--Teissier distribution (KTD). The proposed model extends the KTD to accommodate serial dependence commonly observed in time series while retaining the flexibility of the underlying distribution. The cumulative distribution function (CDF) of the KTD is given by
\begin{eqnarray}\label{ktcdf}
	\F(\x) = 1-\left( 1-\left( 1-e^{\fx} \right)^\a \right)^\b,~~~ \x>0
\end{eqnarray}
where, $\fx =  \lambda\x - e^{\lambda\x} +1 $. $\a,\b>0$ are shape parameters and $\lambda>0$ is scale parameter. 

To introduce temporal dependence into KTD, we adopt the quantile-based parameterization developed by (\cite{mitnik2013kumaraswamy}). Specifically, the shape parameter $b$ is re-parameterized in terms of the conditional $\rho$-th quantile, thereby allowing the time-varying dynamics to be modeled through the conditional quantile rather than through the mean. Under this parameterization,
\[b = \frac{\log(1-\rho)}{\log\left( 1-\left( 1-e^{\fm} \right)^\a \right)}\]
Substituting the above expression into \eqref{ktcdf}, the reparameterized CDF can be written as
\begin{eqnarray}\label{ktrcdf}
	\F(\y; \mu_{\rho}, \a, \lambda) = 1-\left( 1-\left( 1-e^{\fy} \right)^\a \right)^{\frac{\log(1-\rho)}{\log\left( 1-\left( 1-e^{\fm} \right)^\a \right)}},~~~ y>0
\end{eqnarray}
where $\rho\in(0,1)$ is a fixed quantile (assumed known) and $\mu_{\rho}\in(0,1)$ denotes the corresponding $\rho$th quantile. Differentiating the reparameterized distribution function with respect to $y$ yields the corresponding probability density function (PDF),
\begin{eqnarray}\label{ktrpdf}
	\f(\y; \mu_{\rho}, \a, \lambda) = \frac{\a \lambda\log(1-\rho) (e^{\lambda\y} -1) e^{\fy} \left(1-e^{\fy} \right)^{\a-1} \left(  1-\left(1-e^{\fy} \right)^{\a} \right)^{\frac{\log(1-\rho)}{\log\left( 1-\left( 1-e^{\fm} \right)^\a \right)}-1}}{\log\left( 1-\left( 1-e^{\fm} \right)^\a \right)}.
\end{eqnarray}
and the quantile function is given by:
\begin{eqnarray}\label{ktrqf}
 y_{p} (p; \mu_{\rho}, \a, \lambda) = \frac{1}{\lambda}~ ln \left[ -W_{-1} \left(  \left\{ \left(1-(1-p)^{\frac{\log\left( 1-\left( 1-e^{\fm} \right)^\a \right) }{\log(1-\rho)}} \right)^{1/a}-1\right\} /e \right) \right],~~~ 0<p<1
\end{eqnarray}

The proposed KTARMA model is one of the observation-driven time series models, in which the conditional distribution of the response variable evolves over time through a dynamic systematic component. The model extends the generalized autoregressive moving average (GARMA) framework introduced by (\cite{benjamin2003generalized}), where serial dependence is incorporated through autoregressive and moving-average terms acting on a transformed conditional parameter. It also generalizes the $\beta$ARMA model of (\cite{rocha2009beta}), which was developed for bounded continuous data by combining the beta distribution with an ARMA-type dynamic structure. \cite{bayer2017kumaraswamy} proposed the Kumaraswamy autoregressive moving average (KARMA) model by employing a reparameterized Kumaraswamy distribution to model the conditional median of bounded time series. Although the KARMA model provides a flexible alternative to the $\beta$ARMA model, its inference is restricted to the median (\(\rho=0.5\)). Motivated by the increasing interest in quantile-based time series modeling, the proposed KTARMA model extends the median-based framework to a more general conditional quantile setting. Specifically, following the quantile parameterization adopted in the Unit-Weibull ARMA (UWARMA) model of (\cite{pumi2024unit}), the proposed model directly models the conditional $\rho$th quantile, where $\rho\in(0,1)$ is fixed. 

Let $\{Y_t\}_{t\in\mathbb{Z}}$ be a positive-valued stochastic process with support $(0,\infty)$, and let $\{\mathbf{x}_t\}_{t\in\mathbb{Z}}$ denote an $r$-dimensional vector of exogenous covariates. The proposed framework allows considerable flexibility regarding the specification of these covariates, since they may be deterministic, stochastic, or consist of both deterministic and stochastic components. To define the conditional distribution of $Y_t$, it is necessary to specify the information available immediately prior to time $t$ which is obtained by $\mathcal{F}_{t-1}$. The construction of $\mathcal{F}_{t-1}$ depends on the nature of the covariates. For deterministic or predetermined covariates, such as polynomial trends, seasonal indicators, or lagged variables from external processes, the value at time $t$ is already available when forecasting $Y_t$. In contrast, if a covariate is stochastic, its contemporaneous realization is not observed until time $t$, and therefore only its past values are contained in the information set at time $t-1$. To accommodate both situations within a unified framework, the covariate vector is partitioned as \(\mathbf{x}_t=\left(\mathbf{x}_{t}^{d^\top},\mathbf{x}_{t}^{s^\top}\right)^\top,\) where $\mathbf{x}_t^{d}$ is $m_1$-dimensional vector of deterministic (or predetermined) covariates and $\mathbf{x}_t^{s}$ is $m_2$-dimensional vector of stochastic covariates, with $m_1+m_2=m$. Consequently, the filtration available immediately before observing $Y_t$ is defined as
\[\mathcal{F}_{t-1}=\sigma\left\{\mathbf{x}_t^{d},\mathbf{x}_{t-1}^{s},
Y_{t-1},\mathbf{x}_{t-1}^{d},\mathbf{x}_{t-2}^{s},Y_{t-2},\mathbf{x}_{t-2}^{d},
\ldots\right\}\]

Let $\rho\in(0,1)$ be a fixed quantile level. The proposed KTARMA model is constructed by assuming that, conditionally on the information set $\mathcal{F}_{t-1}$, the response variable $Y_t$ follows the reparameterized KTD with conditional quantile parameter $\mu_{\rho,t}$, shape parameter $a$, and scale parameter $\lambda$. That is,
\(Y_t\mid\mathcal{F}_{t-1}\sim
KT(\mu_{\rho,t},a,\lambda),\)
where $a>0$, $\lambda>0$, and $\mu_{\rho,t}$ denotes the conditional $\rho$th quantile of $Y_t$. Consequently, \(P(Y_t\leq\mu_{\rho,t}\mid\mathcal{F}_{t-1})=\rho,\) establishes $\mu_{\rho,t}$ as the dynamic quantile governing the conditional distribution of the process. The temporal evolution of the conditional quantile is introduced through a suitable link function. Let \(g:(0,\infty)\rightarrow\mathbb{R}\) be a known, continuous, and twice continuously differentiable monotone link function. The conditional quantile is connected to the linear predictor through
\begin{eqnarray}\label{linkfun}
	\eta_{\rho,t} = g(\mu_{\rho,t}) = \alpha_{\rho} + {x_t}^\top \beta_{\rho} + \sum_{i=1}^{p} \phi_{\rho,i} [g(y_{t-i}) - {x_{t-i}}^\top \beta_{\rho}] + \sum_{j=1}^{q} \theta_{\rho,j} r_{t-j}
\end{eqnarray}
where $\eta_{\rho,t}$ denotes the linear predictor, $\alpha_{\rho}$ is the intercept parameter, $\beta_{\rho}$ is the vector of regression coefficients associated with the explanatory variables, and $\phi_{\rho}=(\phi_{\rho,1},\ldots,\phi_{\rho,p})^\top$ and $\theta_{\rho}=(\theta_{\rho,1},\ldots,\theta_{\rho,q})^\top$ denote the autoregressive and moving-average parameter vectors, respectively. Several standard link functions can be adopted in the proposed model, including the logit, probit, log--log, and complementary log--log (cloglog) links. 

The error term is recursively defined as \(r_t=g(y_t)-g(\mu_{\rho,t}),\) which measures the deviation of the transformed observation from its corresponding conditional quantile. Since both $\eta_{\rho,t}$ and $\mu_{\rho,t}$ depend only on the information available up to time $t-1$, they are $\mathcal{F}_{t-1}$-measurable. The proposed specification therefore combines the flexibility of the KTD with an ARMA-type dynamic structure, allowing the conditional quantile to evolve over time while accounting for serial dependence in positive-valued observations. The resulting model is referred to as the KTARMA$(p,q)$ model and is completely characterized by the conditional distribution \( Y_t\mid\mathcal{F}_{t-1}\sim KT(\mu_{\rho,t},a,\lambda) \)
along with the dynamic predictor given in \eqref{linkfun}.

\section{Parameter Estimation}\label{cmle}
We estimate the model parameters using the conditional maximum likelihood estimation (CMLE). Let $y_1,y_2,\ldots,y_n$ be a sample from the $\mathrm{KTARMA}(p,q)$ model defined by equation \eqref{ktrpdf} and \eqref{linkfun}, with non-stochastic covariates $x_t\in\mathbb{R}^{r}$. Let \(\gamma_{\rho}
=
\left(
\alpha_{\rho},
\beta_{\rho}^{\top},
\phi_{\rho}^{\top},
\theta_{\rho}^{\top},
a,
\lambda
\right)^{\top}
\in \Omega,\)
where $\Omega \subset \mathbb{R}^{\,r+p+q+3}$ denotes the parameter space. Further, define
\[
A_t = \left(1-e^{\fyt} \right), \qquad 
M_t = \left(1-e^{\fmt} \right), \qquad  
c_t = \frac{\log(1-\rho)}{\log(1-M_t^\a)}
\]
The first $m=\max(p,q)$ observations are treated as fixed initial values and therefore do not contribute to the likelihood. Hence, the conditional log-likelihood function is given by
\[
\ell(\gamma_{\rho})
=
\sum_{t=m+1}^{n}
\ell_t( \mu_{\rho,t}, \a, \lambda),
\]
where,
\begin{eqnarray} \label{MLE}
	\ell_t( \mu_{\rho,t}, \a, \lambda) &=&  \log(\a \lambda) + \log(-\log(1-\rho)) - \log\left\{-\log\left( 1-M_t^\a \right)\right\} + \log(e^{\lambda\y_t} -1) +  \fyt  \nonumber \\ &&  + (\a-1) \log(A_t)  + \left(c_t -1\right) \log \left(1- {A_t}^\a \right)   
\end{eqnarray}

\subsection{Conditional Score Vector}
We differentiate the conditional log-likelihood function given in equation \eqref{MLE} with respect to each unknown parameter in the vector $\gamma_{\rho}$ to get the score vector. Let $\nu_{\rho}= \left(\alpha_{\rho},\beta_{\rho}^{\top},\phi_{\rho}^{\top},\theta_{\rho}^{\top} \right)^\top$, so that $\gamma_{\rho}=\left(\nu_{\rho}^{\top},a,\lambda \right)^{\top}$. We start by computing the analytical partial derivatives of the conditional log-likelihood contribution at time $t$, denoted as $\ell_t(\gamma_{\rho})$, directly with respect to the parameters $a$ and $\lambda$:
\begin{align}
	\frac{\partial \ell_t (\gamma_{\rho})}{\partial a} &= \frac{1}{a} + \log(A_t) + \frac{{M_t}^{a}~ \log(M_t) }{\left( 1-M_t^a \right) \log\left( 1-M_t^a \right)} \left(1 + 	c_t \log \left(1- {A_t}^a \right) \right)  - \left(c_t-1\right) \frac{{A_t}^a ~ \log(A_t) }{ \left(1- {A_t}^a \right) } \\[10pt]
	\frac{\partial \ell_t (\gamma_{\rho})}{\partial \lambda} &= \frac{1}{\lambda} + \frac{e^{\lambda y_t} y_t}{\left(e^{\lambda y_t}-1 \right)} -  y_t \left(e^{\lambda y_t}-1 \right)+ \frac{y_t e^{\fyt} (e^{\lambda y_t} -1)}{A_t} \left[(a-1) - \frac{a \left(c_t-1\right) {A_t}^a}{\left(1- {A_t}^a \right)} \right] \nonumber \\ 
	&\quad + \frac{a \mu_{\rho,t} {M_t}^{a-1} (e^{\lambda \mu_{\rho,t}} -1) e^{\fmt}}{ \left( 1-M_t^a \right)\log\left( 1-M_t^a \right)} \left(1 + 	c_t \log \left(1- {A_t}^a \right) \right)
\end{align}
Next, to compute the score elements for the structural parameters contained within $\nu_{\rho}$, we apply the chain rule because these parameters affect $\ell_t (\gamma_{\rho})$ indirectly through the conditional mean $\mu_{\rho,t}$ and the linear predictor $\eta_{\rho,t}$. Therefore, for any component $\nu_{\rho,j}$, we can write as:
\begin{align} 
	\frac{\partial \ell_t (\gamma_{\rho})}{\partial \nu_{\rho,j}} &= \frac{\partial \ell_t (\gamma_{\rho})}{\partial \mu_{\rho,t}} \frac{d \mu_{\rho,t} }{ d \eta_{\rho,t}} \frac{\partial \eta_{\rho,t}}{\partial \nu_{\rho,j}} 
\end{align}
The first component of the above derivative is given by:
\begin{eqnarray}
	\frac{\partial \ell_t  (\gamma_{\rho})}{\partial \mu_{\rho,t}} &= \frac{a \lambda{M_t}^{a-1} (e^{\lambda\mu_{\rho,t}} -1) e^{\lambda\mu_{\rho,t}} }{\left( 1-M_t^a \right) \log\left( 1-M_t^a \right)} \left(1 + 	c_t \log \left(1- {A_t}^a \right) \right),
\end{eqnarray}
and the second component is:
\begin{align} \label{gmu}
	\frac{d \mu_{\rho,t} }{ d \eta_{\rho,t}}  &= \frac{1}{g'(\mu_{\rho,t})}
\end{align}
where $g'(\cdot)$ denotes the first derivative of the link function. \\
Finally, the last term $\frac{\partial \eta_{\rho,t}}{\partial \nu_{\rho,j}}$ considers the dynamic structure of the $\textrm{KARMA}(p,q)$ process. Differentiating the conditional link function, $ \eta_{\rho,t} = g(\mu_{\rho,t}) $, leads to a recursive system due to the presence of the delayed residuals $ r_{\rho,t-j} = y_{t-j} - \mu_{\rho,t-j} $. Evaluating this derivative in relation to each type of parameter yields:
\begin{align} \label{alpeta}
	\frac{\partial \eta_{\rho,t}}{\partial \alpha_{\rho}} &= 1 + \sum_{j=1}^{q} \theta_{\rho,j} \frac{\partial r_{\rho,t-j}}{\partial \alpha_{\rho}} = 1 - \sum_{j=1}^{q} \theta_{\rho,j} \frac{\partial \eta_{\rho,t-j}}{\partial \alpha_{\rho}}\\[10pt]
	\frac{\partial \eta_{\rho,t}}{\partial \beta_{\rho,l}} &= x_{tl} - \sum_{i=1}^{p} \phi_{\rho,i} x_{(t-i)l} - \sum_{j=1}^{q} \theta_{\rho,j} \frac{\partial \eta_{\rho,t-j}}{\partial \beta_{\rho,l}},\quad l = 1,2,\cdots,r  \\[10pt]
	\frac{\partial \eta_{\rho,t}}{\partial \phi_{\rho,i}} &= g(y_{t-i}) - {x_{t-i}}^\top \beta_{\rho} - \sum_{j=1}^{q} \theta_{\rho,j} \frac{\partial \eta_{\rho,t-j}}{\partial \phi_{\rho,i}},\quad i = 1,2,\cdots,p\\[10pt] \label{theta}
	\frac{\partial \eta_{\rho,t}}{\partial \theta_{\rho,j}} &= r_{\rho,t-j} - \sum_{i=1}^{q} \theta_{\rho,i} \frac{\partial \eta_{\rho,t-i}}{\partial \theta_{\rho,j}},\quad j = 1,2,\cdots,q
\end{align}

To express the total conditional score vector compactly, we stack the time-dependent terms into vector and matrix forms over the effective sample period $t = m+1, \dots, n$, where $m = \max(p, q)$. Let $\mathbf{w}_\mu, \mathbf{w}_a,$ and $\mathbf{w}_\lambda$ be $(n-m) \times 1$ gradient vectors defined as:
\begin{align*}
	\mathbf{w}_\mu &:= \left( \frac{\partial \ell_{m+1}(\gamma_{\rho})}{\partial \mu_{\rho,m+1}}, \cdots, \frac{\partial \ell_n(\gamma_{\rho})}{\partial \mu_{\rho,n}} \right)^\top, \\[8pt]
	\mathbf{w}_a   &:= \left( \frac{\partial \ell_{m+1}(\gamma_{\rho})}{\partial a}, \cdots, \frac{\partial \ell_n(\gamma_{\rho})}{\partial a} \right)^\top, \quad \text{and} \quad
	\mathbf{w}_\lambda := \left( \frac{\partial \ell_{m+1}(\gamma_{\rho})}{\partial \lambda}, \cdots, \frac{\partial \ell_n(\gamma_{\rho})}{\partial \lambda} \right)^\top.
\end{align*}
Furthermore, let $\mathbf{D}_{\nu}$ represent the $(n-m) \times (p + q + r + 1)$ matrix of conditional mean derivatives, where the $(i, j)$-th element corresponds to time index $t = m+i$ and is given by:
\begin{equation*}
	[\mathbf{D}_{\nu}]_{i,j} := \frac{\partial l_{i}(\gamma_{\rho})}{\partial \nu_{\rho,j}}.
\end{equation*}
Using these blocks, the complete conditional score vector $S(\gamma_{\rho})$ can be written as:
\begin{equation*}
	S(\gamma_{\rho}) = \left( S_{\nu_{\rho}}(\gamma_{\rho})^\top, S_a(\gamma_{\rho}), S_{\lambda}(\gamma_{\rho}) \right)^\top,
\end{equation*}
where the individual parameter components are partitioned as:
\begin{equation*}
	S_{\nu_{\rho}}(\gamma_{\rho}) := \mathbf{D}_{\nu}^\top \mathbf{w}_\mu, \quad S_a(\gamma_{\rho}) := \mathbf{1}_{n-m}^\top \mathbf{w}_a, \quad \text{and} \quad S_{\lambda}(\gamma_{\rho}) := \mathbf{1}_{n-m}^\top \mathbf{w}_\lambda,
\end{equation*}
and $\mathbf{1}_{n-m} := (1, \cdots, 1)^\top \in \mathbb{R}^{n-m}$ denotes a column vector of ones.

\subsection{Conditional Information matrix}
This section evaluates the single-observation information matrix $K(\gamma_{\rho})$. Because the unconditional distribution of the KTARMA process cannot be explicitly determined, the classical unconditional Fisher information matrix is unavailable. To construct an equivalent matrix, we apply the approach established by (\cite{kedem2002regression}). This requires evaluating the cumulative conditional information matrix, denoted as \(K(\gamma\rho)\), which is defined by:
\[K(\gamma_{\rho}) = - \sum_{i=m+1}^n  \mathbb{E} \left(\frac{\partial^2 \ell_t (\gamma_{\rho})}{\partial \gamma_{\rho} \partial \gamma_{\rho}^\top} | \mathcal{F}_{t-1} \right)\]

\begin{align*} 
	\frac{\partial^2 \ell_t (\gamma_{\rho})}{\partial \nu_{\rho,i} \partial \nu_{\rho,j}} &= \sum_{t=m+1}^{n} \frac{\partial }{\partial \mu_{\rho,t}}
	\left(\frac{\partial \ell_t (\gamma_{\rho})}{\partial \mu_{\rho,t}} \frac{d \mu_{\rho,t} }{ d \eta_{\rho,t}} \frac{\partial \eta_{\rho,t}}{\partial \nu_{\rho,j}}  \right) \frac{d \mu_{\rho,t} }{ d \eta_{\rho,t}} \frac{\partial \eta_{\rho,t}}{\partial \nu_{\rho,i}}\\
	&= \sum_{t=m+1}^{n} \left\{ \frac{\partial^2 \ell_t (\gamma_{\rho})}{\partial \mu_{\rho,t}^2} \frac{d \mu_{\rho,t} }{ d \eta_{\rho,t}} \frac{\partial \eta_{\rho,t}}{\partial \nu_{\rho,j}} + \frac{\partial \ell_t (\gamma_{\rho})}{\partial \mu_{\rho,t}} \frac{\partial }{\partial \mu_{\rho,t}} \left( \frac{d \mu_{\rho,t} }{ d \eta_{\rho,t}} \frac{\partial \eta_{\rho,t}}{\partial \nu_{\rho,j}}  \right) \right\} \frac{d \mu_{\rho,t} }{ d \eta_{\rho,t}} \frac{\partial \eta_{\rho,t}}{\partial \nu_{\rho,i}}
\end{align*}
From Lemma \ref{lem:lemma1} in \cref{app:lemma}, \(	\mathbb{E} \left(\frac{\partial \ell_t }{\partial \mu_{\rho,t}} | \mathcal{F}_{t-1} \right) = 0\). Therefore,
\begin{align*} 
	\mathbb{E}\left(\frac{\partial^2 \ell_t}{\partial \nu_{\rho,i} \partial \nu_{\rho,j}} | \mathcal{F}_{t-1} \right) 
	&= \mathbb{E}\left(\frac{\partial^2 \ell_t (\gamma_{\rho})}{\partial \mu_{\rho,t}^2} | \mathcal{F}_{t-1} \right) \left\{ \frac{d \mu_{\rho,t} }{ d \eta_{\rho,t}}  \right\}^2  \frac{\partial \eta_{\rho,t}}{\partial \nu_{\rho,i}} \frac{\partial \eta_{\rho,t}}{\partial \nu_{\rho,j}}
\end{align*}
where, \(\left\{ \frac{d \mu_{\rho,t} }{ d \eta_{\rho,t}}  \right\}\) is given by equation \eqref{gmu} and \( \frac{\partial \eta_{\rho,t}}{\partial \nu_{\rho,j}}\) are evaluated using equations \eqref{alpeta}-\eqref{theta}.\\
By applying Lemma \ref{lem:lemma1} on the expectation of \eqref{rho2} evaluated in \cref{app:derivation}, we obtain:
\begin{align*} 
	\mathbb{E}\left(\frac{\partial^2 \ell_t }{\partial \mu_{\rho,t}^2} \;\middle\vert{}\; \mathcal{F}_{t-1} \right)
	&= - \left[\frac{ a \lambda M_t^{a-1} (e^{\lambda \mu_{\rho,t}} - 1) e^{\fmt}  }{(1 - M_t^a) \log(1-M_t^a)}\right]^2
\end{align*}
By using Lemma \ref{lem:lemma1} and \ref{lem:lemma2} on the expectation of equation \eqref{rhoa}, we get:
$$
\begin{aligned}
	\mathbb{E}\left(\frac{\partial^2 l_t}{\partial a \partial \mu_{\rho,t}}  \;\middle\vert{}\; \mathcal{F}_{t-1}\right) &= - \frac{\lambda (e^{\lambda \mu_{\rho,t}} - 1) e^{\fmt}  a c_t M_t^a}{M_t (1 - M_t^a) \left[\log(1 - M_t^a)\right]^2} \left[  \frac{\psi_0(2) - \psi_0(c_t+1)}{a(c_t-1)} \log(1 - M_t^a) +  \frac{M_t^a \log(M_t)}{ c_t (1 - M_t^a)}\right]
\end{aligned}
$$
Hence,
\begin{align*} 
	\mathbb{E}\left(\frac{\partial^2 \ell_t}{\partial a \partial \nu_{\rho,j}} | \mathcal{F}_{t-1} \right) 
	&= 	\mathbb{E}\left(\frac{\partial^2 l_t}{\partial a \partial \mu_{\rho,t}}  \;\middle\vert{}\; \mathcal{F}_{t-1}\right) \left\{ \frac{d \mu_{\rho,t} }{ d \eta_{\rho,t}}  \right\}  \frac{\partial \eta_{\rho,t}}{\partial \nu_{\rho,j}}
\end{align*}
By applying Lemma \ref{lem:lemma1} on the expectation of equation \eqref{rholambda}, first two terms vanish which follows:
$$
\begin{aligned}
	\mathbb{E}\left(\frac{\partial^2 l_t}{\partial \lambda \partial \mu_{\rho,t}}  \;\middle\vert{}\; \mathcal{F}_{t-1}\right) &= \frac{c_t M_t^{a-1} (e^{\lambda \mu_{\rho,t}} - 1) e^{\fmt}}{(1 - M_t^a) \log(1 - M_t^a)} \mathbb{E}\left(\frac{A_t^{a-1} y_t e^{\psi(y_t;\lambda)} (e^{\lambda y_t} -1) }{(1-A_t^a)} \;\middle|\; \mathcal{F}_{t-1}\right)\\ &
	 - \left[\frac{ a \mu_{\rho,t} M_t^{a-1} (e^{\lambda \mu_{\rho,t}} - 1) e^{\fmt}  }{(1 - M_t^a) \log(1-M_t^a)}\right]^2
\end{aligned}
$$
where the conditional expectation on the right-hand side is evaluated in Lemma \ref{lem:lemma4}.\\
Therefore,
\begin{align*} 
	\mathbb{E}\left(\frac{\partial^2 \ell_t}{\partial \lambda \partial \nu_{\rho,j}} | \mathcal{F}_{t-1} \right) 
	&= 	\mathbb{E}\left(\frac{\partial^2 l_t}{\partial \lambda \partial \mu_{\rho,t}}  \;\middle\vert{}\; \mathcal{F}_{t-1}\right) \left\{ \frac{d \mu_{\rho,t} }{ d \eta_{\rho,t}}  \right\}  \frac{\partial \eta_{\rho,t}}{\partial \nu_{\rho,j}}
\end{align*}
By using Lemma \ref{lem:lemma1}, \ref{lem:lemma2} and \ref{lem:lemma3} on the expectation of equation \eqref{a2}, we have:
$$
\begin{aligned}
	\mathbb{E}\left(\frac{\partial^2 l_t}{\partial a^2}  \;\middle\vert{}\; \mathcal{F}_{t-1}\right) &= -\frac{1}{a^2} - \left[  \frac{c_t}{a^2(c_t-2)} \left[ \big(\psi_0(2) - \psi_0(c_t)\big)^2 + \psi_1(2) - \psi_1(c_t) \right] \right]\\
	&\quad - \frac{ 2 c_t M_t^a \log M_t }{1 - M_t^a) \log(1 - M_t^a)} \cdot \left[ \frac{\psi_0(2) - \psi_0(c_t+1)}{a(c_t-1)}  \right] - \left[ \frac{M_t^a \log M_t}{1 - M_t^a) \log(1 - M_t^a)}  \right]^2
\end{aligned}
$$
The expectation of equation \eqref{alambda} is given by:
$$
\begin{aligned}
	\mathbb{E}\left(\frac{\partial^2 l_t}{\partial \lambda \partial a }  \;\middle\vert{}\; \mathcal{F}_{t-1}\right) &= \mathbb{E}\left(\frac{y_t e^{\psi(y_t;\lambda)} (e^{\lambda y_t} -1) }{A_t} \;\middle|\; \mathcal{F}_{t-1}\right) - \frac{a \mu_{\rho,t} M_t^{2a-1} \log M_t  (e^{\lambda \mu_{\rho,t}} - 1) e^{\fmt}}{(1 - M_t^a)^2 \log^2(1 - M_t^a)} \\  &\quad - \frac{c_t M_t^{a-1} \mu_{\rho,t}(e^{\lambda \mu_{\rho,t}} - 1)e^{\fmt}}{(1 - M_t^a)\log(1 - M_t^a)} \frac{\psi_0(2) -  \psi_0(c_t+1)}{(c_t-1)} \\  &\quad - a (c_t -1) \mathbb{E}\left(\frac{A_t^{a-1} y_t e^{\psi(y_t;\lambda)} (e^{\lambda y_t} -1) \log(A_t)}{(1-A_t^a)^2} \;\middle|\; \mathcal{F}_{t-1}\right) \\  &\quad - a \left[ \frac{c_t M_t^a \log M_t}{(1 - M_t^a)\log(1 - M_t^a)} - (c_t-1) \right]
	\mathbb{E}\left(\frac{A_t^{a-1} y_t e^{\psi(y_t;\lambda)} (e^{\lambda y_t} -1) }{(1-A_t^a)} \;\middle|\; \mathcal{F}_{t-1}\right) 
\end{aligned}
$$
where, the first term is computed in Lemma \ref{lem:lemma5}, second and third term is obtained via Lemma \ref{lem:lemma1} and \ref{lem:lemma2}, the fourth term is evaluated using Lemma \ref{lem:lemma6} and the final term is established in Lemma \ref{lem:lemma4}.\\
The expectation of equation \eqref{lam2} is given by:
$$
\small
\begin{aligned} 
	&\mathbb{E}\left(\frac{\partial^2 l_t}{\partial \lambda^2}  \;\middle\vert{}\; \mathcal{F}_{t-1}\right) = -\frac{1}{\lambda^2} - \mathbb{E}\left(\frac{y_t^2 e^{\lambda y_t}}{(e^{\lambda y_t} -1)^2} \;\middle|\; \mathcal{F}_{t-1}\right)  - \mathbb{E}\left(y_t^2 e^{\lambda y_t} \;\middle|\; \mathcal{F}_{t-1} \right)  \\&  + (a-1) \Bigg[
	\mathbb{E}\left(\frac{ y_t^2 e^{\psi(y_t;\lambda)} e^{\lambda y_t} }{A_t} \;\middle|\; \mathcal{F}_{t-1}\right)  - \mathbb{E}\left(\frac{ y_t^2 e^{\psi(y_t;\lambda)} (e^{\lambda y_t} -1)^2  }{A_t} \;\middle|\; \mathcal{F}_{t-1}\right)  \\& - \mathbb{E}\left(\frac{ y_t^2 e^{2\psi(y_t;\lambda)} (e^{\lambda y_t} -1)^2  }{A_t^2} \;\middle|\; \mathcal{F}_{t-1}\right)
	\Bigg] - a (c_t-1) \Bigg\{
	\mathbb{E}\left(\frac{A_t^{a-1} y_t^2 e^{\lambda y_t}  e^{\psi(y_t;\lambda)} }{(1-A_t^a)} \;\middle|\; \mathcal{F}_{t-1}\right) 	\\& - \mathbb{E}\left(\frac{A_t^{a-1} y^2 (e^{\lambda y_t} -1)^2 e^{\psi(y_t;\lambda)} }{(1-A_t^a)} \;\middle|\; \mathcal{F}_{t-1}\right) - \mathbb{E}\left(\frac{A_t^{a-2} y_t^2 (e^{\lambda y_t} -1)^2 e^{2 \psi(y_t;\lambda)} }{(1-A_t^a)} \;\middle|\; \mathcal{F}_{t-1}\right)
	\Bigg \} \\ & - a^2 (c_t-1) \mathbb{E}\left(\frac{A_t^{a-2} y_t^2 (e^{\lambda y_t} -1)^2 e^{2 \psi(y_t;\lambda)} }{(1-A_t^a)^2} \;\middle|\; \mathcal{F}_{t-1}\right) - \frac{ a^2 \mu_t c_t M_t^{a-1}(1-M_t)(e^{\lambda \mu_t}-1)}{(1-M_t^a)\log(1-M_t^a)} \mathbb{E}\left(\frac{A_t^{a-1} y_t e^{\psi(y_t;\lambda)} (e^{\lambda y_t} -1) }{(1-A_t^a)} \;\middle|\; \mathcal{F}_{t-1}\right)
	\\& + \frac{ a^2 \mu_t M_t^{a-1} (e^{\lambda \mu_t} - 1) e^{\psi(\mu_t;\lambda)} }{(1-M_t^a) \log(1-M_t^a)}
	\Bigg[
	\frac{- \mu_t M_t^{a-1}(1-M_t)(e^{\lambda \mu_t}-1) }{(1-M_t^a) \log(1-M_t^a)}
	- c_t \mathbb{E}\left(\frac{A_t^{a-1} y_t e^{\psi(y_t;\lambda)} (e^{\lambda y_t} -1) }{(1-A_t^a)} \;\middle|\; \mathcal{F}_{t-1}\right)
	\Bigg] 
\end{aligned}
$$
where, the above expectation terms are derived explicitly in Lemmas \ref{lem:lemma1}, \ref{lem:lemma4}, and \ref{lem:lemma7}--\ref{lem:lemma15}.

The conditional information matrix for $\gamma_{\rho}$ is given by:
\[K(\gamma_{\rho}) = \begin{pmatrix}
	K_{a,a} & K_{a, \lambda} & K_{a, \nu_{\rho}}\\
	K_{\lambda, a } & K_{\lambda, \lambda} & K_{\lambda, \nu_{\rho} } \\
	K_{\nu_{\rho}, a} & K_{\nu_{\rho}, \lambda} & K_{\nu_{\rho}, \nu_{\rho}}
\end{pmatrix}\]
where, \(K_{a,a} = - \mathbb{E}\left(\frac{\partial^2 l_t}{\partial a^2}  \;\middle\vert{}\; \mathcal{F}_{t-1}\right)\); \(K_{a,\lambda} = K_{\lambda, a} = - \mathbb{E}\left(\frac{\partial^2 l_t}{\partial \lambda \partial a}  \;\middle\vert{}\; \mathcal{F}_{t-1}\right)\); \(K_{\lambda, \lambda} = - \mathbb{E}\left(\frac{\partial^2 l_t}{\partial \lambda^2}  \;\middle\vert{}\; \mathcal{F}_{t-1}\right)\); \(K_{\nu_{\rho},\nu_{\rho}} = - \mathbf{D}_{\nu}^\top \mathbb{E}\left(\frac{\partial^2 \ell_t }{\partial \mu_{\rho,t}^2} \;\middle\vert{}\; \mathcal{F}_{t-1} \right) \frac{1}{(g'(\mu_{\rho,t}))^2} \mathbf{D}_{\nu}  \); \(K_{\nu_{\rho},a} = K^\top_{a, \nu_{\rho}} = - \mathbf{D}_{\nu}^\top \mathbb{E}\left(\frac{\partial^2 \ell_t }{\partial a \partial \mu_{\rho,t}} \;\middle\vert{}\; \mathcal{F}_{t-1} \right) \frac{1}{g'(\mu_{\rho,t})} \) ; \(K_{\nu_{\rho}, \lambda} = K^\top_{\lambda, \nu_{\rho}} = - \mathbf{D}_{\nu}^\top \mathbb{E}\left(\frac{\partial^2 \ell_t }{\partial \lambda \partial \mu_{\rho,t}} \;\middle\vert{}\; \mathcal{F}_{t-1} \right) \frac{1}{g'(\mu_{\rho,t})} \).

\section{Model Specification, Diagnostics and Forecasting} \label{diagnostic}
This section describes the model selection, diagnostic analysis, and forecasting adopted for the proposed KTARMA model. The procedure consists of three main steps: selection of statistically relevant predictive regressors, identification of the optimal quantile level, and assessment of model adequacy through residual diagnostics and out-of-sample forecasting. Initially, all the predictive regressors are incorporated into the KTARMA model. The statistical significance of the regression coefficients is subsequently examined and predictors that do not make a significant contribution to the model are excluded based on their $p-$values. The KTARMA model is then refitted using the retained regressors to preserve the statistically relevant information contained in the covariates. Following the variable-selection procedure, the model is also estimated over different quantile levels. 

The adequacy of the selected KTARMA model is assessed through residual diagnostics. Residuals serve as a critical metric for verifying whether a fitted model provides a robust approximation of the data distribution (\cite{kedem2002regression}). While traditional diagnostics often rely on standardized Pearson's or deviance residuals, non-Gaussian and bounded frameworks can introduce non-linear distortions. Consequently, this study utilizes randomized quantile residuals (\cite{dunn1996randomized}), which offer substantial theoretical advantages over conventional residual types. The residuals for the KTARMA framework are mathematically defined as follows:
\begin{equation}
	r_{\rho,t}=\Phi ^{-1}\left(F_{\mu _{\rho, t}}({y}_{t}\mid \mathcal{F}_{t-1})\right)
\end{equation}
where \(\Phi ^{-1}\) denotes the standard normal quantile function. A primary advantage of quantile residuals is that, under correct model specification, their empirical distribution converges to an approximate standard normal distribution (\(\mathcal{N}(0,1)\)). Accordingly, the index time-series plot of these quantile residuals should display a random, homoscedastic scatter patterns. Furthermore, when a model is correctly specified, the residuals display true white noise behavior, following a zero-mean, constant-variance, and uncorrelated stochastic process. To formally evaluate the adequacy of the model and confirm the complete removal of serial dependencies, a Ljung–Box test (\cite{ljung1978measure}) is deployed directly upon the residual series.

The predictive validation of the proposed KTARMA model is executed with an out-of-sample forecasting framework. Let $\hat{\boldsymbol{\gamma}}_{\rho}$ denote the parameter vector obtained via CMLE based on the historical training sample $\{y_1, \dots, y_n\}$ with its associated covariates $\mathbf{x}_1, \dots, \mathbf{x}_n$. Based on these estimated parameters, the $h$-step ahead out-of-sample forecasts, denoted by $\{\widehat{y}_{n+1}, \dots, \widehat{y}_{n+h}\}$, are derived sequentially. For a target forecast horizon $h$, the predictive engine is mathematically defined as:
\begin{equation}  
	\widehat{y}_{n+h} := g^{-1} \left( \widehat{\alpha}_{\rho} + \mathbf{x}_{n+h}^\top \widehat{\boldsymbol{\beta}}_{\rho} + \sum_{i=1}^{p} \widehat{\phi}_{\rho,i} \left( g([y_{n+h-i}]^*) - \mathbf{x}_{n+h-i}^\top \widehat{\boldsymbol{\beta}}_{\rho} \right) + \sum_{j=1}^{q} \widehat{\theta}_{\rho,j} \widehat{r}_{\rho,n+h-j} \right)  
\end{equation}
where $g(\cdot)$ represents the link function and \( [y_t]^* := y_t \mathbb{I}(1 \leq t \leq n) + \widehat{y}_t \mathbb{I}(t \geq n + 1) \) with $\mathbb{I}(\cdot)$ is the standard indicator function. The internal innovation error terms during the transition phases are captured recursively by mapping the observations through the link space:
\begin{equation}  
	\widehat{r}_{\rho,t} = g(y_t) - g(\widehat{\mu}_{\rho,t})  
\end{equation}
In the presence of covariates, the execution of the $h$-step ahead forecast equation requires that future values for the covariates ($\mathbf{x}_{n+h}$) across the independent testing period be explicitly provided.

\section{Simulation Study} \label{simulation}
This section evaluates the finite-sample performance of the CMLE formulated in Section~\ref{cmle} via a comprehensive Monte Carlo simulation study of the KTARMA model. The simulation framework generates 1000 independent random samples for each configuration with varying sample sizes of $n = 50, 100, 200, 500 $. Two different parameter scenarios are considered to test the estimators in different structural dynamic conditions:
\begin{itemize}
	\item \textbf{Scenario 1 [KTARMA(1,1) with two covariates]:} 
	$a = 0.9$, $\lambda = 1.5$, $\alpha = 0.4$, $\beta_1 = -0.6$, $\beta_2 = 0.8$, $\phi_1 = 0.25$, $\theta_1 = -0.4$.
	
	\item \textbf{Scenario 2 [KTARMA(2,2) with one covariate]:} 
	 $a = 1.2$, $\lambda = 0.8$, $\alpha = 0.7$, $\beta_1 = 0.4$, $\phi_1 = 0.4$, $\phi_2 = -0.12$, $\theta_1 = -0.2$, $\theta_2 = 0.6$.
\end{itemize}
\begin{table}[!ht] 
	\centering
	\small
	\caption{Mean estimates, Bias and MSE of KTARMA model for $a = 0.9$, $\lambda = 1.5$, $\alpha = 0.4$, $\beta_1 = -0.6$, $\beta_2 = 0.8$, $\phi_1 = 0.25$, $\theta_1 = -0.4$}
	\resizebox{\textwidth}{!}{
		\begin{tabular}{@{}lcccccccccccc@{}}
			\toprule
			& \multicolumn{4}{c}{Mean} & \multicolumn{4}{c}{Bias} & \multicolumn{4}{c}{MSE} \\ 
			\cmidrule(lr){2-5} \cmidrule(lr){6-9} \cmidrule(lr){10-13}
			{} & 50 & 100 & 200 & 500 &
			50 & 100 & 200 & 500 &
			50 & 100 & 200 & 500 \\ 
			\midrule
			
			\multicolumn{13}{c}{$\rho = 0.25$} \\
			\midrule
			
			$\hat{a}$ 
			& 0.9916 & 0.9373 & 0.9218 & 0.9082 
			& 0.0916 & 0.0373 & 0.0218 & 0.0082
			& 0.0483 & 0.0220 & 0.0109 & 0.0047 \\
			
			$\hat{\lambda}$ 
			& 1.5824 & 1.5923 & 1.5338 & 1.5093
			& 0.0824 & 0.0923 & 0.0338 & 0.0093
			& 0.1597 & 0.1277 & 0.0853 & 0.0423 \\
			
			$\hat{\alpha}$ 
			& 0.4735 & 0.4512 & 0.4298 & 0.4109
			& 0.0735 & 0.0512 & 0.0298 & 0.0109
			& 0.0759 & 0.0458 & 0.0224 & 0.0113 \\
			
			$\hat{\beta}_1$ 
			& -0.6099 & -0.6040 & -0.6049 & -0.5923
			& -0.0099 & -0.0040 & -0.0049 & 0.0077
			& 0.1425 & 0.0921 & 0.0527 & 0.0214 \\
			
			$\hat{\beta}_2$ 
			& 0.8425 & 0.8388 & 0.8004 & 0.7986
			& 0.0425 & 0.0388 & 0.0004 & -0.0014
			& 0.1862 & 0.0974 & 0.0472 & 0.0228 \\
			
			$\hat{\phi}_1$ 
			& 0.2719 & 0.2648 & 0.2573 & 0.2528
			& 0.0219 & 0.0148 & 0.0073 & 0.0028
			& 0.0403 & 0.0151 & 0.0061 & 0.0025 \\
			
			$\hat{\theta}_1$ 
			& -0.4329 & -0.4237 & -0.4107 & -0.4044
			& -0.0329 & -0.0237 & -0.0107 & -0.0044
			& 0.0448 & 0.0167 & 0.0068 & 0.0027 \\
			
			\midrule
			\multicolumn{13}{c}{$\rho = 0.5$} \\
			\midrule
			
			$\hat{a}$ 
			& 0.9875 & 0.9358 & 0.9189 & 0.9073
			& 0.0875 & 0.0358 & 0.0189 & 0.0073
			& 0.0417 & 0.0191 & 0.0101 & 0.0046 \\
			
			$\hat{\lambda}$ 
			& 1.5580 & 1.5789 & 1.5304 & 1.5063
			& 0.0580 & 0.0789 & 0.0304 & 0.0063
			& 0.1895 & 0.1609 & 0.1202 & 0.0743 \\
			
			$\hat{\alpha}$ 
			& 0.4068 & 0.4063 & 0.4090 & 0.4019
			& 0.0068 & 0.0064 & 0.0090 & 0.0019
			& 0.0650 & 0.0427 & 0.0221 & 0.0118 \\
			
			$\hat{\beta}_1$ 
			& -0.6022 & -0.6023 & -0.6032 & -0.5912
			& -0.0022 & -0.0023 & -0.0032 & 0.0088
			& 0.1663 & 0.1181 & 0.0725 & 0.0302 \\
			
			$\hat{\beta}_2$ 
			& 0.8476 & 0.8349 & 0.7990 & 0.7974
			& 0.0476 & 0.0349 & -0.0010 & -0.0026
			& 0.2213 & 0.1316 & 0.0664 & 0.0317 \\
			
			$\hat{\phi}_1$ 
			& 0.2654 & 0.2721 & 0.2550 & 0.2565
			& 0.0154 & 0.0221 & 0.0050 & 0.0065
			& 0.0719 & 0.0276 & 0.0104 & 0.0040 \\
			
			$\hat{\theta}_1$ 
			& -0.4267 & -0.4301 & -0.4080 & -0.4080
			& -0.0267 & -0.0301 & -0.0080 & -0.0080
			& 0.0740 & 0.0291 & 0.0107 & 0.0042 \\
			
			\midrule
			\multicolumn{13}{c}{$\rho = 0.75$} \\
			\midrule
			
			$\hat{a}$ 
			& 0.9806 & 0.9329 & 0.9169 & 0.9073
			& 0.0806 & 0.0329 & 0.0169 & 0.0073
			& 0.0352 & 0.0160 & 0.0087 & 0.0043 \\
			
			$\hat{\lambda}$ 
			& 1.5600 & 1.5845 & 1.5311 & 1.4976
			& 0.0600 & 0.0845 & 0.0311 & -0.0024
			& 0.2132 & 0.1882 & 0.1557 & 0.1053 \\
			
			$\hat{\alpha}$ 
			& 0.3403 & 0.3560 & 0.3961 & 0.3934
			& -0.0597 & -0.0440 & -0.0039 & -0.0066
			& 0.0898 & 0.0664 & 0.0371 & 0.0196 \\
			
			$\hat{\beta}_1$ 
			& -0.5943 & -0.5886 & -0.6005 & -0.5875
			& 0.0057 & 0.0114 & -0.0005 & 0.0125
			& 0.1811 & 0.1408 & 0.0915 & 0.0432 \\
			
			$\hat{\beta}_2$ 
			& 0.8223 & 0.8412 & 0.7913 & 0.7949
			& 0.0223 & 0.0412 & -0.0087 & -0.0051
			& 0.2555 & 0.1632 & 0.0919 & 0.0459 \\
			
			$\hat{\phi}_1$ 
			& 0.2761 & 0.2725 & 0.2489 & 0.2557
			& 0.0261 & 0.0225 & -0.0011 & 0.0057
			& 0.1302 & 0.0721 & 0.0403 & 0.0132 \\
			
			$\hat{\theta}_1$ 
			& -0.4495 & -0.4373 & -0.4030 & -0.4077
			& -0.0495 & -0.0373 & -0.0030 & -0.0077
			& 0.1322 & 0.0711 & 0.0400 & 0.0128 \\
			
			\bottomrule
	\end{tabular}}
	\label{sim1}
\end{table}
\begin{table}[!ht] 
	\centering
	\small
	\caption{Mean estimates, Bias and MSE of KTARMA model for $a = 1.2$, $\lambda = 0.8$, $\alpha = 0.7$, $\beta_1 = 0.4$, $\phi_1 = 0.4$, $\phi_2 = -0.12$, $\theta_1 = -0.2$, $\theta_2 = 0.6$}
	\resizebox{\textwidth}{!}{
		\begin{tabular}{@{}lcccccccccccc@{}}
			\toprule
			& \multicolumn{4}{c}{Mean} & \multicolumn{4}{c}{Bias} & \multicolumn{4}{c}{MSE} \\ 
			\cmidrule(lr){2-5} \cmidrule(lr){6-9} \cmidrule(lr){10-13}
			{} & 50 & 100 & 200 & 500 &
			50 & 100 & 200 & 500 &
			50 & 100 & 200 & 500 \\ 
			\midrule
			
			\multicolumn{13}{c}{$\rho = 0.25$} \\
			\midrule
			
			$\hat{a}$ 
			& 1.2778 & 1.2386 & 1.2213 & 1.2132
			& 0.0778 & 0.0386 & 0.0213 & 0.0132
			& 0.0328 & 0.0175 & 0.0104 & 0.0053 \\
			
			$\hat{\lambda}$ 
			& 0.7801 & 0.7551 & 0.7568 & 0.7526
			& -0.0199 & -0.0449 & -0.0432 & -0.0474
			& 0.1107 & 0.0913 & 0.0677 & 0.0429 \\
			
			$\hat{\alpha}$ 
			& 0.7699 & 0.7524 & 0.7248 & 0.7244
			& 0.0699 & 0.0524 & 0.0248 & 0.0244
			& 0.1788 & 0.1399 & 0.0948 & 0.0494 \\
			
			$\hat{\beta}_1$ 
			& 0.4575 & 0.4271 & 0.4202 & 0.4051
			& 0.0575 & 0.0271 & 0.0202 & 0.0051
			& 0.2657 & 0.2176 & 0.1205 & 0.0428 \\
			
			$\hat{\phi}_1$ 
			& 0.4236 & 0.4057 & 0.4040 & 0.4011
			& 0.0236 & 0.0057 & 0.0040 & 0.0011
			& 0.0293 & 0.0071 & 0.0024 & 0.0007 \\
			
			$\hat{\phi}_2$ 
			& -0.0357 & -0.0782 & -0.0993 & -0.1143
			& 0.0843 & 0.0418 & 0.0207 & 0.0057
			& 0.1106 & 0.0389 & 0.0125 & 0.0037 \\
			
			$\hat{\theta}_1$ 
			& -0.2040 & -0.1986 & -0.1997 & -0.1994
			& -0.0040 & 0.0014 & 0.0003 & 0.0006
			& 0.0163 & 0.0045 & 0.0017 & 0.0005 \\
			
			$\hat{\theta}_2$ 
			& 0.6042 & 0.5981 & 0.5972 & 0.6006
			& 0.0042 & -0.0019 & -0.0028 & 0.0006
			& 0.0160 & 0.0055 & 0.0019 & 0.0005 \\
			
			\midrule
			\multicolumn{13}{c}{$\rho = 0.5$} \\
			\midrule
			
			$\hat{a}$ 
			& 1.3328 & 1.2512 & 1.2199 & 1.2083
			& 0.1328 & 0.0512 & 0.0199 & 0.0083
			& 0.0590 & 0.0212 & 0.0090 & 0.0040 \\
			
			$\hat{\lambda}$ 
			& 0.9137 & 0.8675 & 0.8524 & 0.8041
			& 0.1137 & 0.0675 & 0.0524 & 0.0041
			& 0.1961 & 0.1703 & 0.1448 & 0.0936 \\
			
			$\hat{\alpha}$ 
			& 0.7605 & 0.7418 & 0.7092 & 0.7160
			& 0.0605 & 0.0418 & 0.0092 & 0.0160
			& 0.1184 & 0.0781 & 0.0392 & 0.0162 \\
			
			$\hat{\beta}_1$ 
			& 0.4544 & 0.4147 & 0.4104 & 0.4033
			& 0.0544 & 0.0147 & 0.0104 & 0.0034
			& 0.2030 & 0.1487 & 0.0621 & 0.0187 \\
			
			$\hat{\phi}_1$ 
			& 0.4218 & 0.4017 & 0.4020 & 0.3991
			& 0.0218 & 0.0017 & 0.0020 & -0.0009
			& 0.0178 & 0.0043 & 0.0016 & 0.0005 \\
			
			$\hat{\phi}_2$ 
			& -0.0903 & -0.1052 & -0.1168 & -0.1145
			& 0.0297 & 0.0148 & 0.0032 & 0.0055
			& 0.0415 & 0.0106 & 0.0045 & 0.0016 \\
			
			$\hat{\theta}_1$ 
			& -0.2072 & -0.1990 & -0.2007 & -0.1988
			& -0.0072 & 0.0010 & -0.0007 & 0.0012
			& 0.0073 & 0.0017 & 0.0007 & 0.0002 \\
			
			$\hat{\theta}_2$ 
			& 0.6005 & 0.6008 & 0.6018 & 0.5996
			& 0.0005 & 0.0008 & 0.0018 & -0.0004
			& 0.0069 & 0.0019 & 0.0008 & 0.0003 \\
			
			\midrule
			\multicolumn{13}{c}{$\rho = 0.75$} \\
			\midrule
			
			$\hat{a}$ 
			& 1.3346 & 1.2540 & 1.2176 & 1.2020
			& 0.1346 & 0.0540 & 0.0176 & 0.0020
			& 0.0543 & 0.0200 & 0.0077 & 0.0033 \\
			
			$\hat{\lambda}$ 
			& 0.9293 & 0.9362 & 0.9410 & 0.9047
			& 0.1293 & 0.1362 & 0.1410 & 0.1047
			& 0.2361 & 0.2374 & 0.2333 & 0.1929 \\
			
			$\hat{\alpha}$ 
			& 0.7092 & 0.7022 & 0.6951 & 0.7048
			& 0.0092 & 0.0022 & -0.0049 & 0.0048
			& 0.0902 & 0.0544 & 0.0208 & 0.0084 \\
			
			$\hat{\beta}_1$ 
			& 0.4554 & 0.4150 & 0.4015 & 0.4041
			& 0.0554 & 0.0150 & 0.0015 & 0.0041
			& 0.1689 & 0.1275 & 0.0478 & 0.0167 \\
			
			$\hat{\phi}_1$ 
			& 0.3951 & 0.4100 & 0.4015 & 0.4010
			& -0.0049 & 0.0100 & 0.0015 & 0.0010
			& 0.0699 & 0.0245 & 0.0072 & 0.0021 \\
			
			$\hat{\phi}_2$ 
			& -0.1106 & -0.1318 & -0.1216 & -0.1191
			& 0.0094 & -0.0118 & -0.0016 & 0.0009
			& 0.0743 & 0.0262 & 0.0095 & 0.0031 \\
			
			$\hat{\theta}_1$ 
			& -0.1953 & -0.2076 & -0.2019 & -0.2005
			& 0.0047 & -0.0076 & -0.0019 & -0.0005
			& 0.0391 & 0.0107 & 0.0021 & 0.0006 \\
			
			$\hat{\theta}_2$ 
			& 0.6265 & 0.6133 & 0.6057 & 0.6019
			& 0.0265 & 0.0133 & 0.0057 & 0.0019
			& 0.0257 & 0.0072 & 0.0023 & 0.0006 \\
			
			\bottomrule
	\end{tabular}}
	\label{sim2}
\end{table}
These scenarios are systematically evaluated at three conditional quantiles, $\rho \in \{0.25, 0.50, 0.75\}$. Random samples for the underlying KTARMA($p,q$) processes are generated by inverting the conditional distribution via the quantile function specified in equation~\eqref{ktrqf}. Given the conditional nature of the estimation framework, the dynamic recursive components must be appropriately initialized. Let $m = \max(p, q)$ denote the order of the model. For the initial $m$ periods ($t = 1, 2, \dots, m$), the residuals are initialized to their unconditional expectation, $r_{\rho,t} = 0$, while the conditional mean is set to $\mu_{\rho,t} = g^{-1}(\alpha_{\rho} + x_t^\top \beta_{\rho})$. For the subsequent periods ($t = m+1, \dots, n$), the conditional means are updated according to the full structural specification $\mu_{\rho,t} = g^{-1}(\eta_{\rho,t})$. 

The optimization is executed using R programming software. To ensure numerical stability, the true parameter values are utilized as the initial guess. The performance of the CMLE is evaluated based on three standard empirical metrics: the mean estimates, empirical biases, and mean squared errors (MSEs). The simulation results for the first and second parameter scenarios are presented in Table~\ref{sim1} and \ref{sim2}, respectively. Across both scenarios and all three evaluated quantiles, a consistent pattern emerges: as the sample size $n$ increases, the empirical mean estimates converge toward their respective true parameter values. Simultaneously, both the absolute biases and MSEs exhibit a monotonic decline with increasing sample sizes. This behavior in bias and  empirically verifies the asymptotic consistency and efficiency of the derived estimators in finite samples.

\section{Real life Application} \label{application}
\subsection{Data Acquired}
In this study, we consider monthly gridded rainfall over the NWH—encompassing Jammu \& Kashmir (JK), Himachal Pradesh (HP), and Uttarakhand (UK)—for the 25-year period from 2001 to 2025. Obtained from the India Meteorological Department (IMD) (\small \url{https://www.imdpune.gov.in/cmpg/Griddata/Rainfall_25_NetCDF.html}) \normalsize with a spatial resolution of 0.25°×0.25° (\cite{pai2014development}), this dataset is considered as the dependent variable. 
\begin{table}[!ht]
	\centering
	\caption{Datasets and explanatory variables used in the rainfall modelling analysis.}
	\label{tab:data_sources}
	\begin{tabular}{ccc}
		\hline
		 \textbf{Variables} & Pressure levels & \textbf{Data source}  \\
		\hline
		Surface temperature (T) & - & ERA5\\
		Specific humidity (SH) & 850 hPa & ERA5  \\
		Geopotential height (Z) & 200 and 500 hPa & ERA5 \\
		Zonal wind (U) & 200, 500 and 850 hPa & ERA5\\
		Meridional wind (V) & 200, 500 and 850 hPa & ERA5\\
		 North Atlantic Oscillation (NAO)& - & NOAA  \\
		Arctic Oscillation (AO) &- & NOAA  \\
		Southern Oscillation Index (SOI) & -	& NOAA  \\
		\hline
	\end{tabular}
\end{table}
To account for atmospheric and climatic influences on NWH rainfall, the model incorporates local meteorological variables and large-scale teleconnection indices as explanatory variables. Temperature and specific humidity represent local thermodynamic and moisture conditions, while wind components and geopotential height capture atmospheric circulation and moisture transport. The Southern Oscillation Index (SOI), North Atlantic Oscillation (NAO), and Arctic Oscillation (AO) represent large-scale climatic variability. The rationale for including these variables is discussed in the Introduction.

Table~\ref{tab:data_sources} presents the explanatory variables considered in the rainfall modelling analysis, along with their respective data sources and spatial resolutions for the time period 2001--2025. The pressure-level variables are obtained from the ERA5 pressure-level monthly means product, available through the Copernicus Climate Change Service (C3S) Climate Data Store (CDS) at
(\small \url{https://cds.climate.copernicus.eu/datasets/reanalysis-era5-pressure-levels-monthly-means}), \normalsize
while surface temperature is obtained from the ERA5 single-level monthly means product at
(\small \url{https://cds.climate.copernicus.eu/datasets/reanalysis-era5-single-levels-monthly-means}) with spatial resolution of  $0.25^\circ \times 0.25^\circ$.
\normalsize
The Arctic Oscillation (AO), North Atlantic Oscillation (NAO), and Southern Oscillation Index (SOI) are obtained from the National Oceanic and Atmospheric Administration (NOAA) for the same period and are available at
(\small \url{https://www.cpc.ncep.noaa.gov/products/precip/CWlink/}). \normalsize

\subsection{Conversion of Gridded Rainfall Data into Self-Organizing Map Zones}
With a spatial resolution of $0.25^\circ \times 0.25^\circ$, the NWH region comprises 525 spatial grids. Fitting the KTARMA model separately to all grids increases computational complexity and may introduce redundancy due to similar rainfall behaviour across neighbouring grids. Therefore, the grids are aggregated into relatively homogeneous rainfall zones to reduce dimensionality while retaining the main spatial variability.
\begin{figure}[H]
	\begin{center}
		\includegraphics[width=1\textwidth]{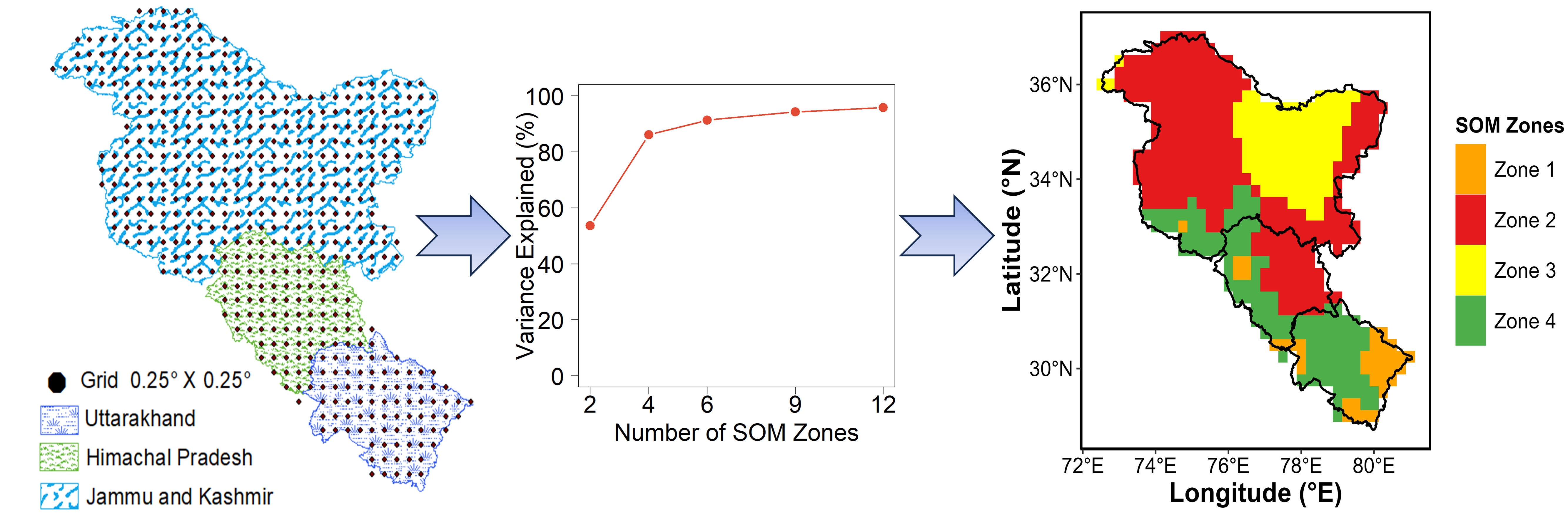}
	\end{center}
	\caption{Gridded Rainfall Data into Self-Organizing Map Zones}
	\label{fig:som_zones}
\end{figure} 
Several clustering techniques can be considered for this purpose. K-Means clusters the observations into a fixed number of groups based on the similarity within a cluster (\cite{forgy1965cluster}). Hierarchical Clustering gives a nested view of the relationships between observations (\cite{sokal1958statistical}). DBSCAN identifies clusters based on local data density and can accommodate irregularly shaped groups while distinguishing isolated observations as noise (\cite{ester1996density}). In contrast, Self-Organizing Maps (SOM) apply an unsupervised neural-network architecture to map complicated, high-dimensional and nonlinear patterns onto an organized lower-dimensional space while keeping the similarity structure of the data (\cite{kohonen1990self}). To obtain these homogeneous zones, SOM are chosen since the method is able to capture the non-linear atmospheric dynamics and produce an orderly topological grid (\cite{philippopoulos2014performance}). 

SOM represents observations through prototype vectors arranged on a low-dimensional map, with similar observations assigned to neighbouring neurons. For each grid, five rainfall characteristics are considered: mean rainfall, standard deviation, proportion of zero-rainfall months, annual seasonal amplitude, and semi-annual seasonal amplitude, with the latter two obtained from harmonic regression with periods of 12 and 6 months, respectively. The characteristics are standardized before training. SOM configurations with 2, 4, 6, 9, and 12 nodes are evaluated using the percentage of variance explained, with the elbow criterion selecting a $2\times2$ hexagonal SOM and four rainfall zones (Figure~\ref{fig:som_zones}). The rainfall and atmospheric variables are then averaged within each zone to obtain representative time series for subsequent modelling.

\subsection{KTARMA Modelling and Out-of-Sample Forecasting across SOM Zones}
To evaluate the performance of the proposed KTARMA model across the SOM-based rainfall zones, each monthly time series is divided into a training period from 2001 to 2024, consisting of $n=288$ observations, and a testing period corresponding to the year 2025, which is reserved for out-of-sample forecasting and model comparison. Figure~\ref{ACFSOM} presents the sample autocorrelation functions (ACF) for the four SOM zones, indicating the presence of both serial dependence and pronounced seasonal behaviour in the rainfall series.  
\begin{figure}[!ht]
	\begin{minipage}[t]{0.48\textwidth}
		\centering
		\includegraphics[width=1\textwidth]{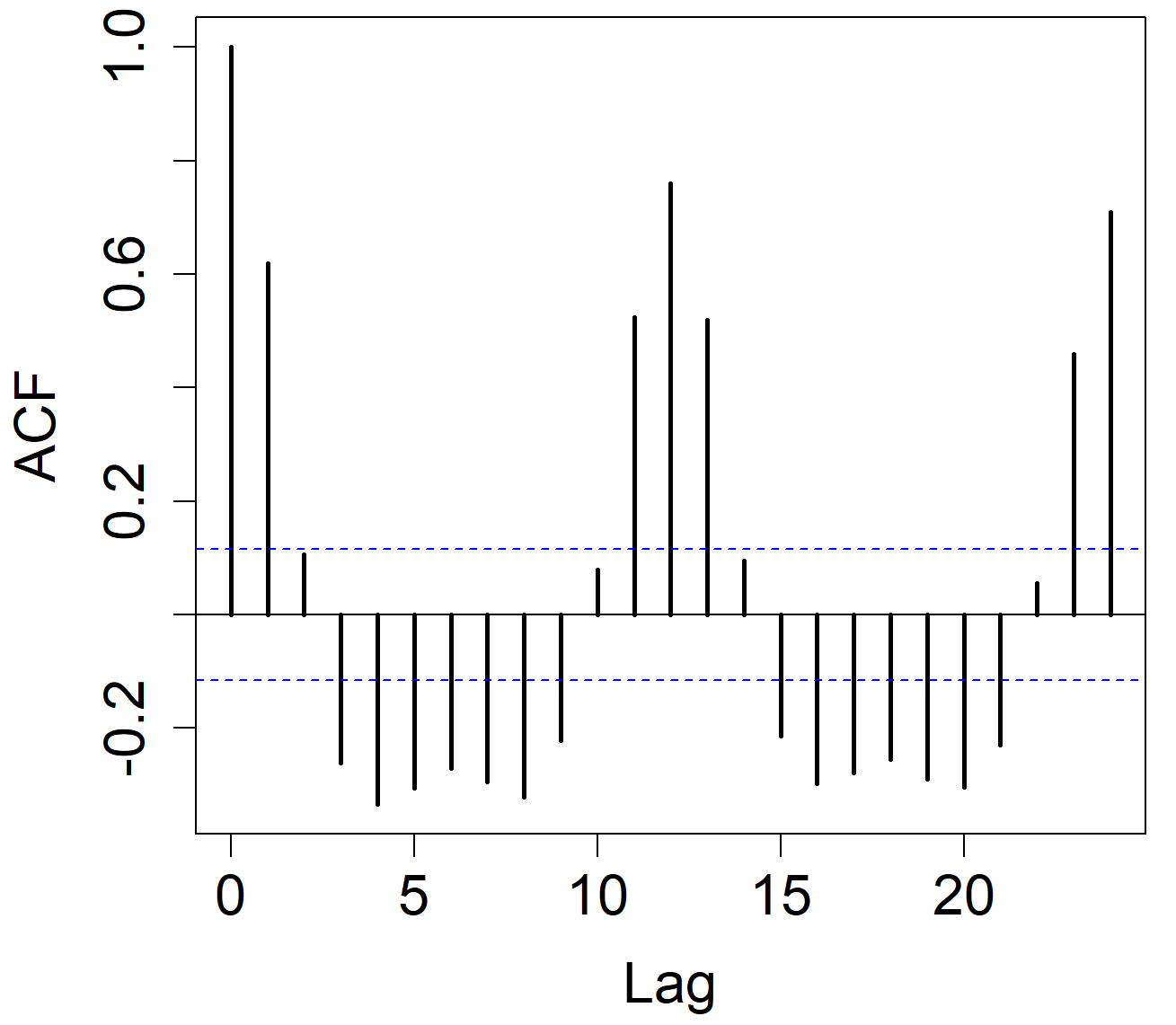}
		\caption*{(a) Zone 1 }
	\end{minipage}\hfill
	\begin{minipage}[t]{0.48\textwidth}
		\centering
		\includegraphics[width=1\textwidth]{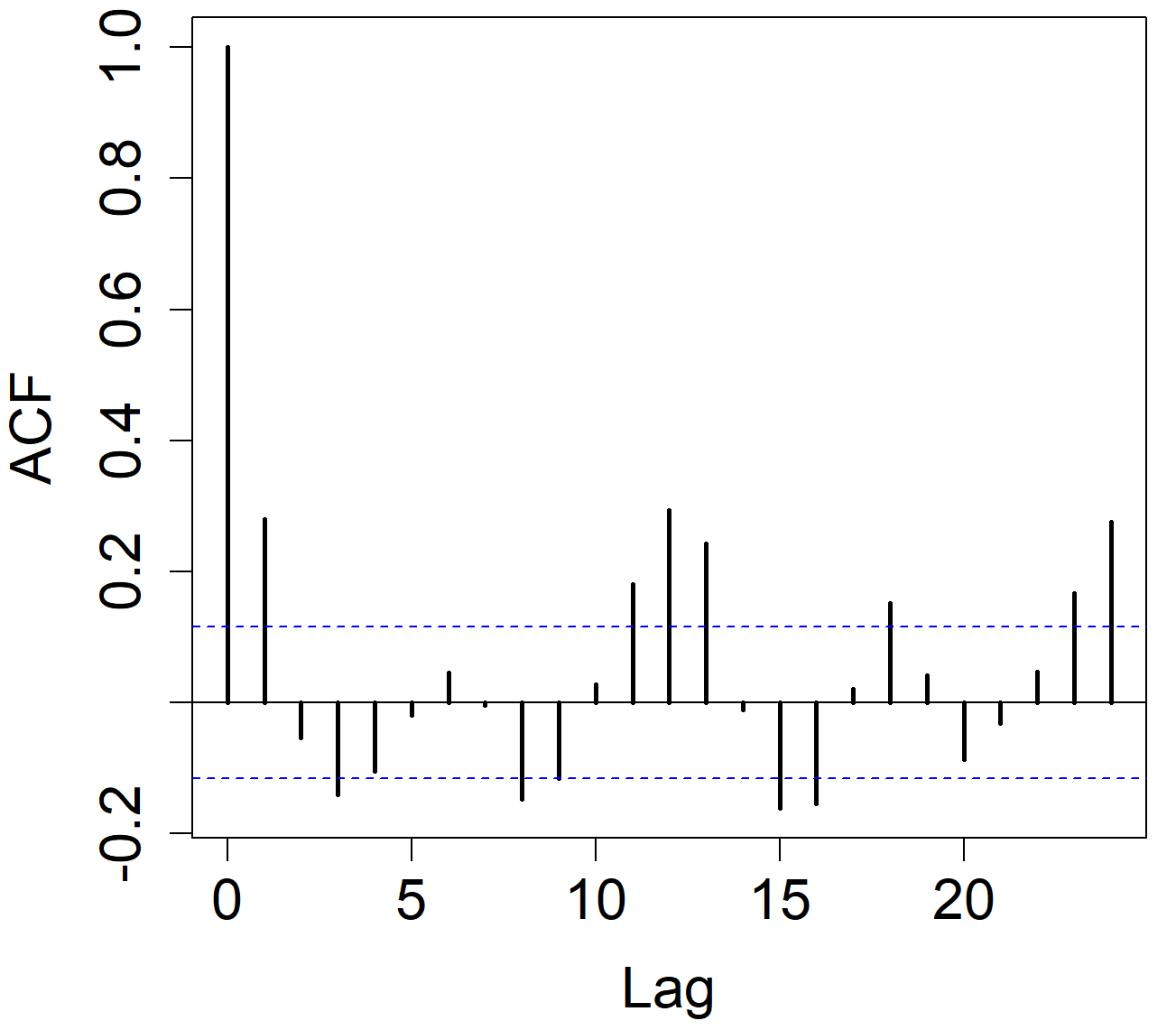}
		\caption*{(b) Zone 2}
	\end{minipage}\hfill
	\begin{minipage}[t]{0.48\textwidth}
		\centering
		\includegraphics[width=1\textwidth]{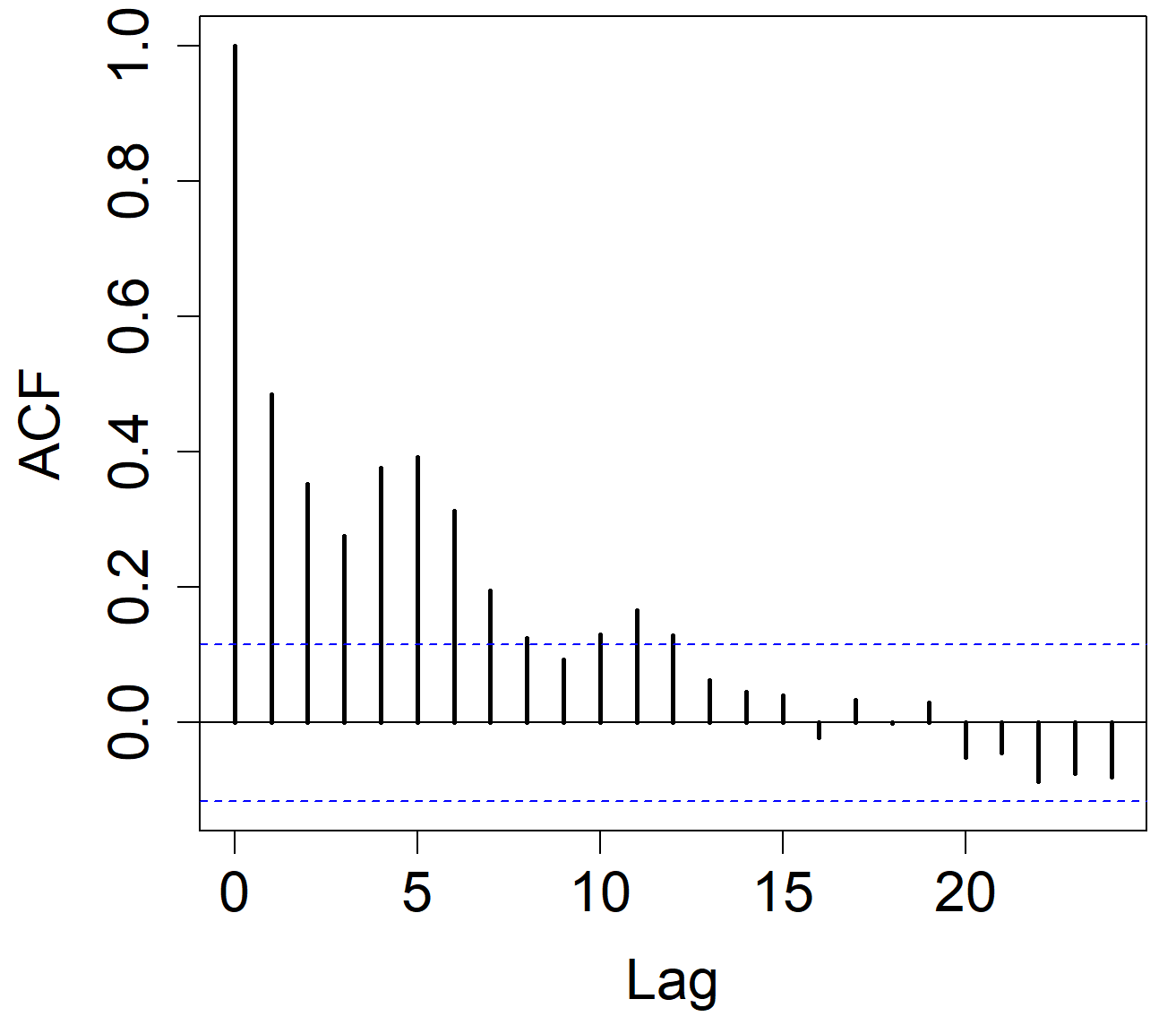}
		\caption*{(c) Zone 3}
	\end{minipage}\hfill
	\begin{minipage}[t]{0.48\textwidth}
		\centering
		\includegraphics[width=1\textwidth]{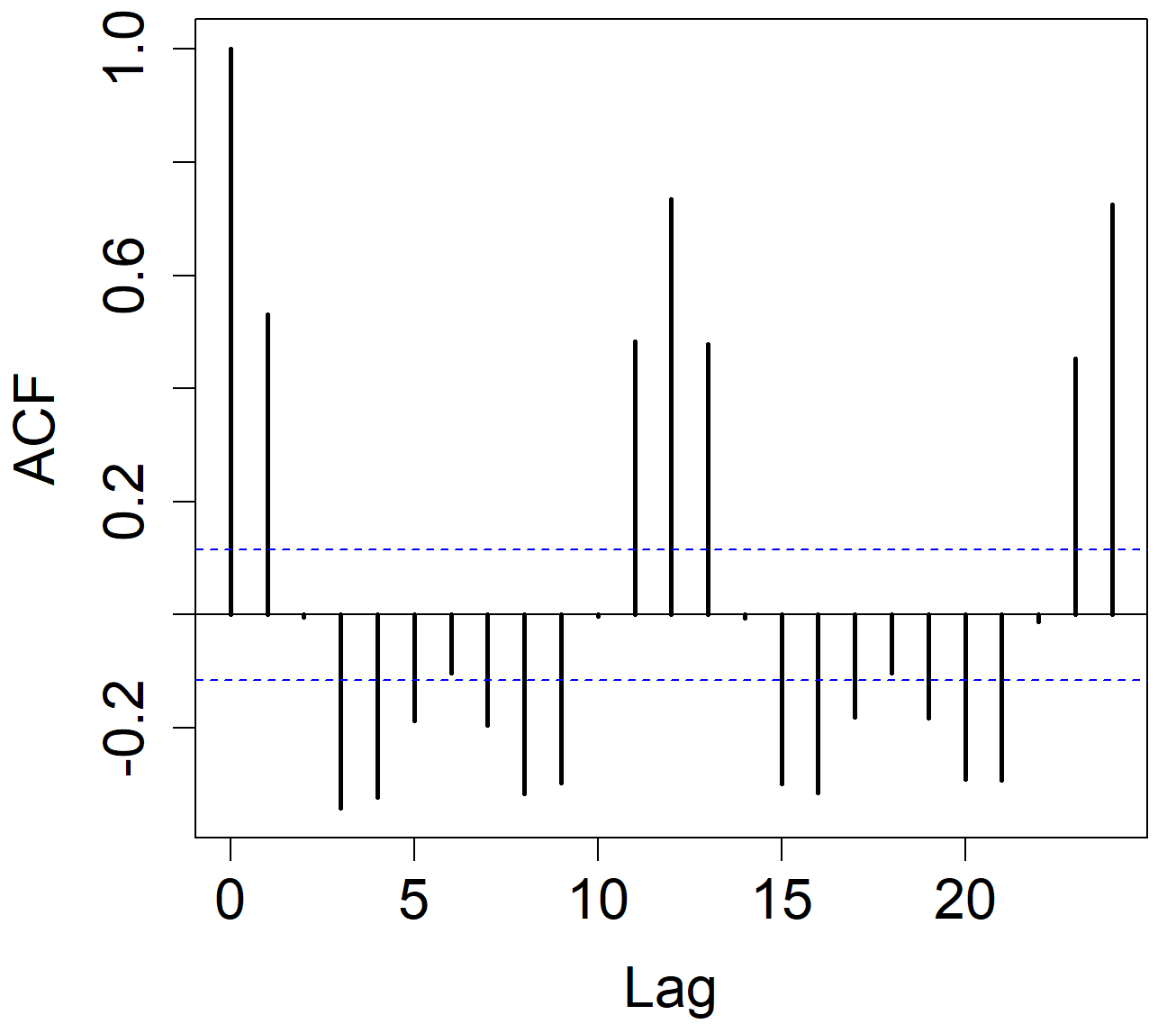}
		\caption*{(d) Zone 4}
	\end{minipage}
	\caption{ACF of SOM Zones}
	\label{ACFSOM}
\end{figure}

In addition, the atmospheric variables listed in Table~\ref{tab:data_sources} are spatially averaged over the grids belonging to each SOM zone and included as potential explanatory variables. Since lagged large-scale climate oscillations have also been reported to influence NWH rainfall (\cite{mishra2026spatiotemporal}), the lagged values of the AO, NAO, and SOI indices are also considered as candidate regressors. The same set of candidate regressors is considered for all four SOM zones; however, the results report only those covariates that show a statistically significant effect on rainfall within each respective zone. 

Since the proposed KTARMA model is formulated for a general conditional $\rho^{th}$ quantile rather than being restricted to the conditional median, the model is fitted separately across a range of quantile levels. 
\begin{figure}[!ht]
	\begin{center}
		\includegraphics[width=0.5\textwidth]{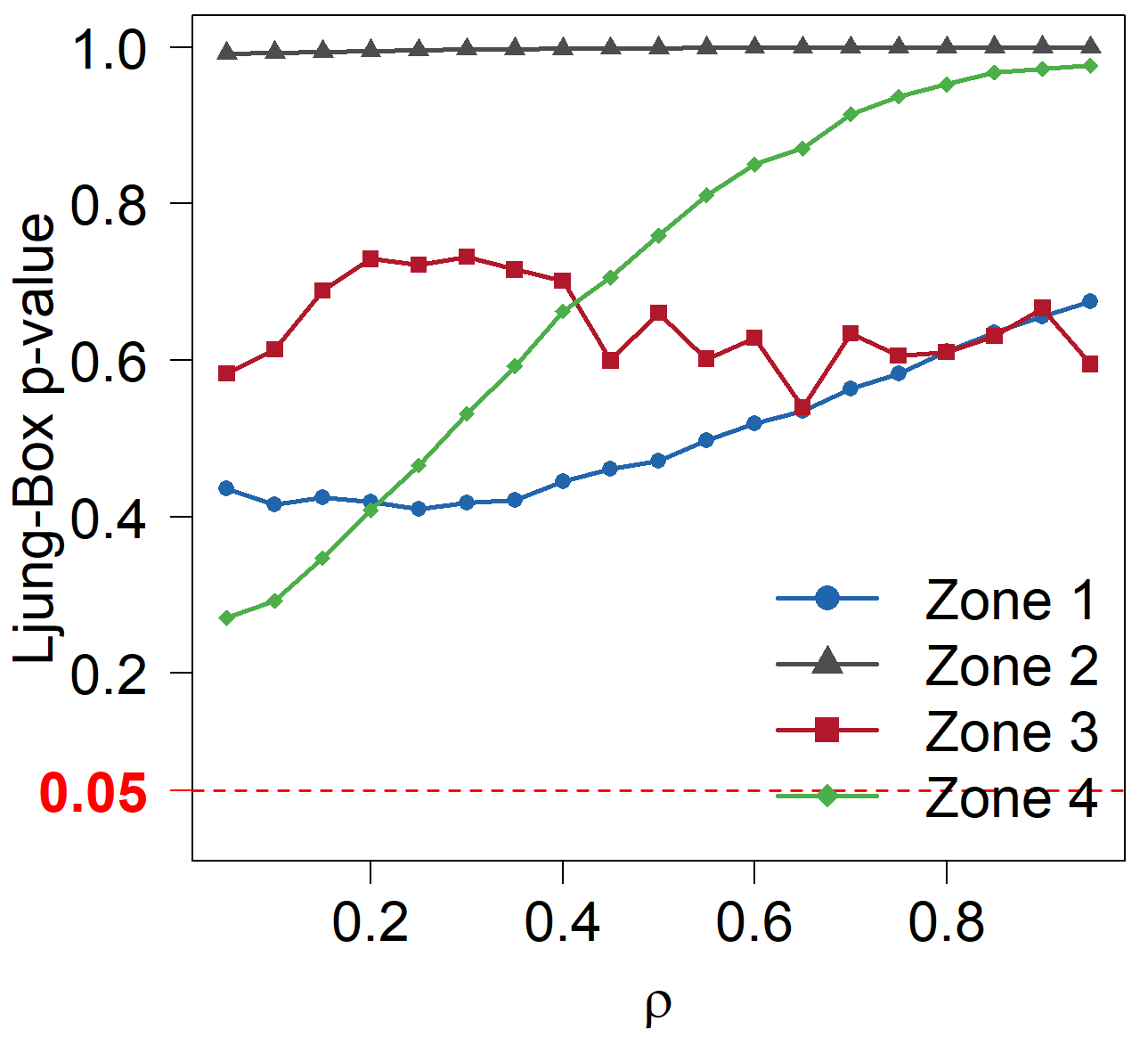}
	\end{center}
	\caption{Ljung--Box test $p$-values at lag 20 for the residuals of fitted KTARMA models across the SOM zones and different values of $\rho$.}
	\label{ljungtest}
\end{figure} 
Specifically, for each of the four SOM-derived rainfall zones, the KTARMA model is fitted separately at the quantile levels $\rho=0.05,0.10,\ldots,0.95$, resulting in 19 quantile-specific models for each zone. The same set of candidate covariates is considered across all quantile levels, while the statistically significant autoregressive and moving-average orders, $p$ and $q$, identified for each zone are retained in the corresponding quantile-specific models. A logit link function is employed to relate the conditional $\rho^{th}$ quantile, $\mu_{\rho,t}$, to the linear predictor. Accordingly, the systematic component of the KTARMA model is specified as
\begin{eqnarray}\label{logit}
	\log \left( \frac{\mu_{\rho,t}}{1-\mu_{\rho,t}} \right)
	= \alpha_{\rho} + {x_t}^\top \beta_{\rho}
	+ \sum_{i=1}^{p} \phi_{\rho,i}
	[g(y_{t-i}) - {x_{t-i}}^\top \beta_{\rho}]
	+ \sum_{j=1}^{q} \theta_{\rho,j} r_{t-j}.
\end{eqnarray}
\begin{figure}[!ht]
	\begin{minipage}[t]{0.48\textwidth}
		\centering
		\includegraphics[width=1\textwidth]{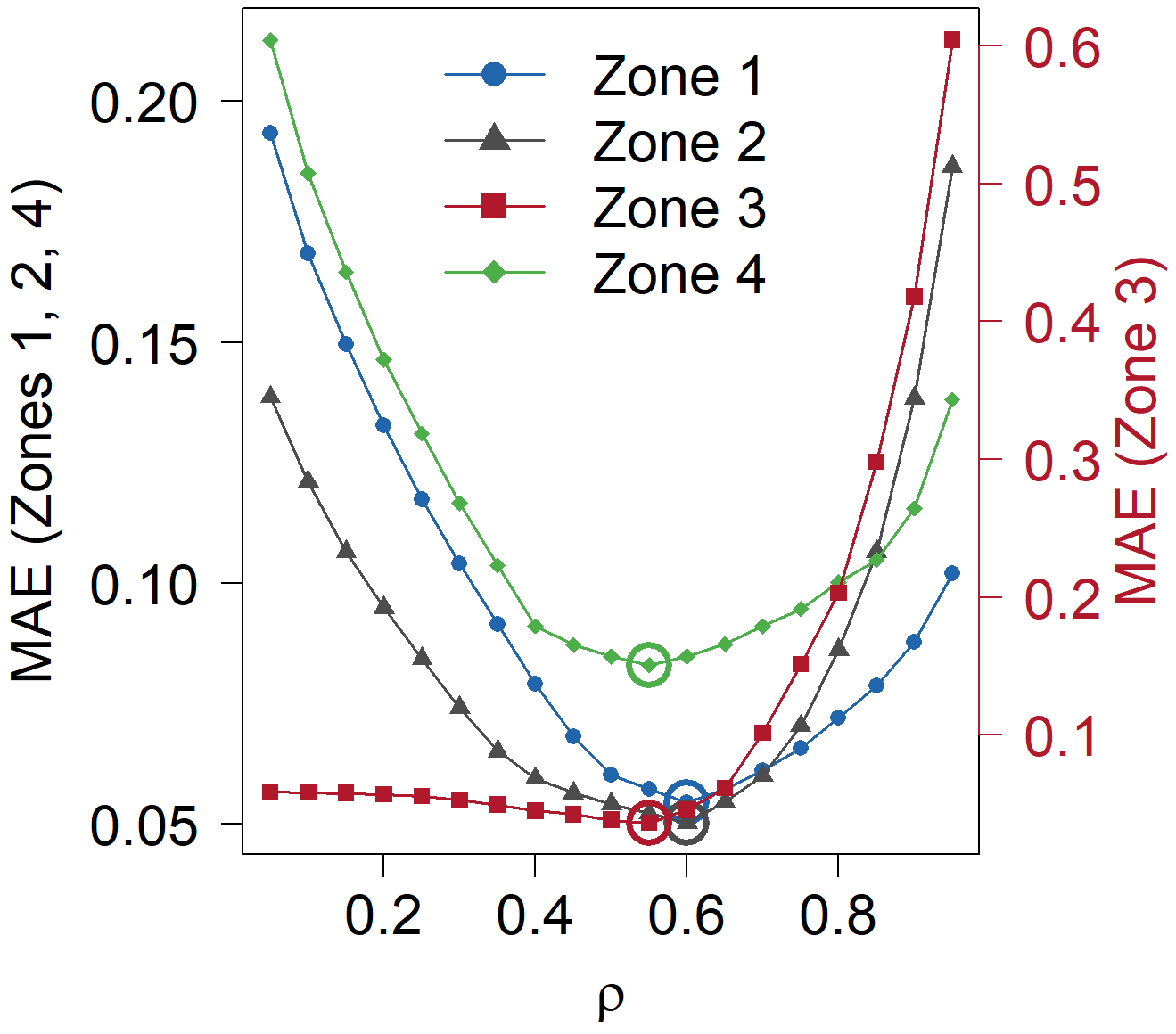}
		\caption*{(a) }
	\end{minipage}\hfill
	\begin{minipage}[t]{0.48\textwidth}
		\centering
		\includegraphics[width=1\textwidth]{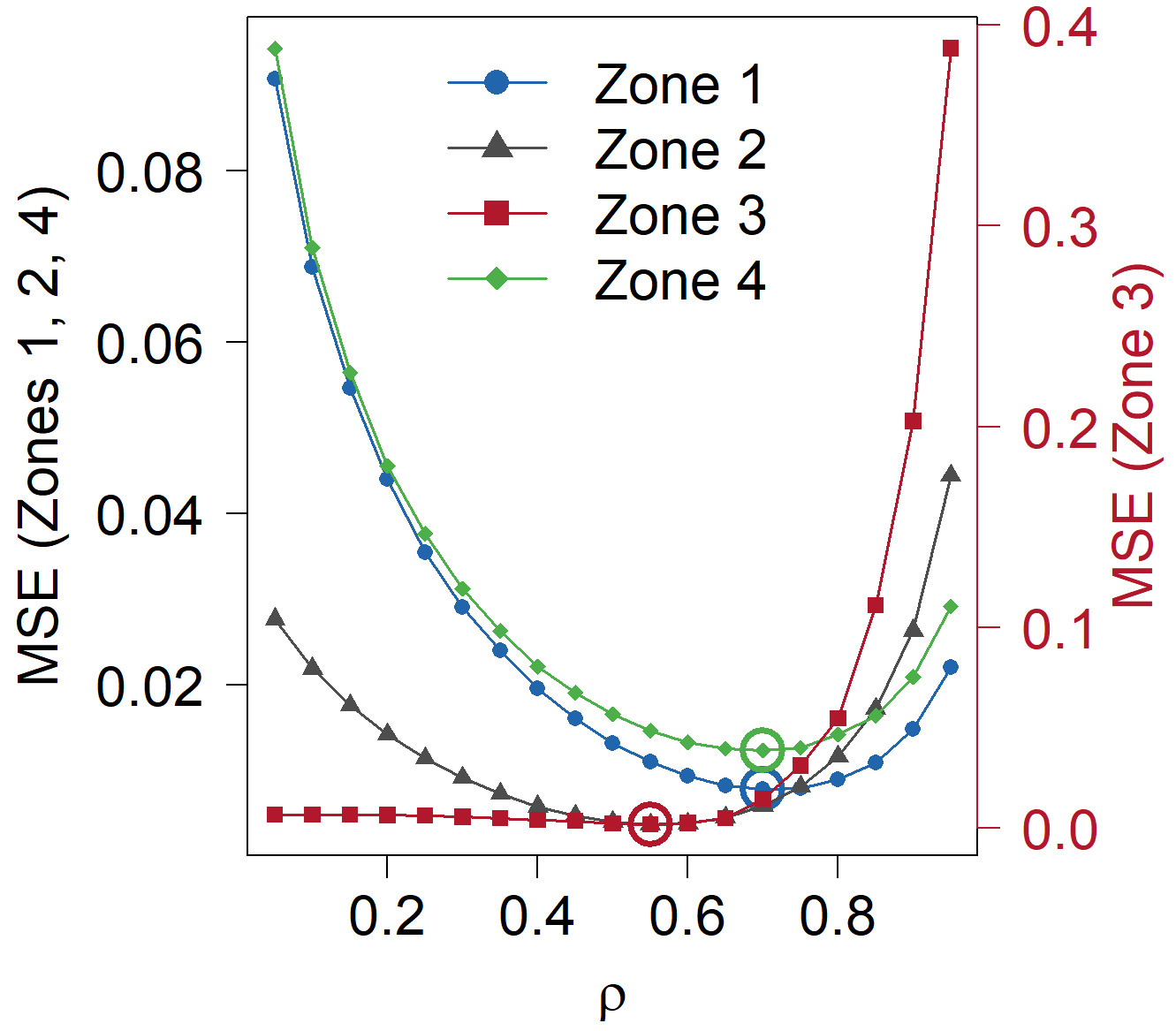}
		\caption*{(b) }
	\end{minipage}
	\caption{Out-of-sample forecasting performance of the KTARMA model across different quantile levels for the SOM-derived rainfall zones, measured by (a) MAE and (b) MSE.}
	\label{MSEMAE}
\end{figure}

The Ljung--Box $p$-values for lag=20 for the fitted models across the 19 quantile levels and four SOM zones are presented in Figure~\ref{ljungtest}. In all four zones, the $p$-values remain above the 0.05 significance level across the considered quantiles, indicating that there is no significant evidence of residual serial correlation. This provides support for the adequacy of the fitted KTARMA models in capturing the temporal dependence structure of rainfall across the different quantile levels.

Furthermore, the forecasting performance of the fitted KTARMA models is evaluated through out-of-sample forecasts for the year 2025 across the considered quantile levels. The forecasting accuracy is assessed using the Mean Absolute Error (MAE) and Mean Squared Error (MSE), with the corresponding results presented in Figure~\ref{MSEMAE}. Since the magnitude of the MSE values for Zone 3 is comparatively larger than that of the other zones, a separate scale is used for the corresponding axis to clearly illustrate the variation in forecasting errors across all zones. The results indicate that the minimum MSE is attained at $\rho=0.70$ for Zones 1 and 4, whereas $\rho=0.55$ yields the minimum MSE for Zones 2 and 3. 
Accordingly, the quantile level associated with the minimum MSE in each zone is selected as the preferred quantile for forecasting. The corresponding parameter estimates and inferential results of the selected KTARMA models are reported in Table~\ref{tableKTARMA}. The estimated KTARMA models reveal distinct rainfall responses across the four SOM-derived zones. 
\begin{table}[!ht]
	\centering
	\footnotesize 
	\renewcommand{\arraystretch}{0.85} 
	\caption{Fitted KTARMA Model Parameter Estimates Across All Four Zones}
	\label{tableKTARMA}
		\begin{tabular}{lcccc}
			\toprule
			\toprule
			Parameters & Zone 1 ($\rho = 0.7$) & Zone 2 ($\rho = 0.55$) & Zone 3 ($\rho = 0.55$) & Zone 4 ($\rho = 0.7$) \\ 
			\midrule
			$a$       & 0.6977*** (0.0401)  & 0.8950*** (0.0508)  & 0.2503*** (0.0163)  & 0.8882*** (0.0538)  \\
			$\lambda$ & 3.4565*** (0.0106)  & 1.6562*** (0.3149)  & 1.3795*** (0.2937)  & 3.4647*** (0.0096)  \\
			$\alpha$  & -1.9433*** (0.0477) & -1.3211*** (0.0788) & -0.8461*** (0.2926) & -1.9334*** (0.2457) \\
			\midrule
			$\beta_{\sin(2 \pi t/12)}$ & -1.0699*** (0.2999) & -1.7105*** (0.2385) & -2.1440*** (0.5160) & -0.6042** (0.2579)  \\
			$\beta_{\cos(2 \pi t/12)}$ & -3.3852*** (0.4049) & -2.6257*** (0.2615) & -1.9344*** (0.5028) & -2.0962*** (0.3605) \\
			$\beta_{\sin(2 \pi t/6)}$  & -0.1718 (0.1689)   & 0.4281*** (0.0731)  & 0.6288*** (0.1851)  & 0.2175 (0.1297)     \\
			$\beta_{\cos(2 \pi t/6)}$  & -0.6114*** (0.0886) & -0.3545*** (0.0733) & -0.6896*** (0.1820) & -0.3795*** (0.0749) \\
			\midrule
			$\beta_{\mathrm{AO-lagged}}$  & 0.1072** (0.0441)  & —                   & —                   & —                   \\
			$\beta_{\mathrm{NAO-lagged}}$ & —                   & —                   & -0.3207*** (0.1080) & —                   \\
			$\beta_{\mathrm{T}}$       & -2.5244*** (0.3809) & —                   & —                   & -1.0932*** (0.3184) \\
			$\beta_{\mathrm{GH500}}$      & —                   & -0.8252*** (0.1776) & -1.3212*** (0.4240) & —                   \\
			$\beta_{\mathrm{SH850}}$      & 1.5872*** (0.3628)  & —                   & —                   & 0.9753*** (0.2651)  \\
			$\beta_{\mathrm{U200}}$       & -0.5017*** (0.2305) & -0.5875*** (0.1022) & —                   & —                   \\
			$\beta_{\mathrm{U500}}$       & 0.5454*** (0.1675)  & 0.5286*** (0.1212)  & —                   & —                   \\
			$\beta_{\mathrm{U850}}$       & -0.2869** (0.0882)  & -0.8541*** (0.1355) & -1.0575*** (0.2313) & -0.6089*** (0.0831) \\
			$\beta_{\mathrm{V200}}$       & —                   & -0.2344*** (0.0721) & —                   & —                   \\
			$\beta_{\mathrm{V500}}$       & 0.2759*** (0.0629)  & 0.7730*** (0.1447)  & —                   & 0.1553*** (0.0513)  \\
			\midrule
			$\phi_{1}$                    & —                   & 0.1101*** (0.0380)  & 0.7359*** (0.0830)  & -0.3461** (0.1612)  \\
			$\theta_{1}$                  & —                   & —                   & -0.5502*** (0.1070) & 0.4468*** (0.1560)  \\
			\bottomrule
			\bottomrule
		\end{tabular}
	\begin{flushleft}
		\centering
		Standard errors are reported in parentheses. *** $p < 0.01$, ** $p < 0.05$
	\end{flushleft}
\end{table}
The annual harmonic components are statistically significant in all four zones, while only cosine component of the semi-annual  is significant. The semi-annual sine component is significant for Zones 2, 3 but is not statistically significant for Zone 1 and 4. For Zone 1 (lower UK) and Zone 4 (remaining UK and lower HP), rainfall is inversely affected by temperature and lower-tropospheric zonal wind ($U_{850}$), but directly stimulated by specific humidity ($\mathrm{SP}_{850}$). Mid-tropospheric winds ($U_{500}, V_{500}$) contribute positively to Zone 1, while only $V_{500}$ is significant for Zone 4. Additionally, Zone 1 exhibits a significant positive association with the lagged Arctic Oscillation ($\mathrm{AO}$) and a negative association with upper-level zonal wind ($U_{200}$). No significant AR or MA terms are retained, indicating the absence of significant residual serial dependence in the fitted model. Conversely, Zone 4 exhibits strong temporal dependence, governed by a negative AR ($\phi_1$) and a positive MA ($\theta_1$) parameter.

In the high-altitude western regions, Zone 2 (JK and upper HP) is suppressed by mid-tropospheric geopotential height ($\mathrm{GH}_{500}$), lower-to-upper zonal winds ($U_{850}, U_{200}$), and upper meridional winds ($V_{200}$), but relates positively to mid-level winds ($U_{500}, V_{500}$) and a positive AR process. Finally, for Zone 3 (Ladakh region), rainfall is significantly reduced by the lagged North Atlantic Oscillation ($\mathrm{NAO}$), $\mathrm{GH}_{500}$, and low-level zonal winds ($U_{850}$) along with a positive AR and a negtive MA coefficient.
\begin{figure}[!ht]
	\begin{minipage}[t]{0.48\textwidth}
		\centering
		\includegraphics[width=1\textwidth]{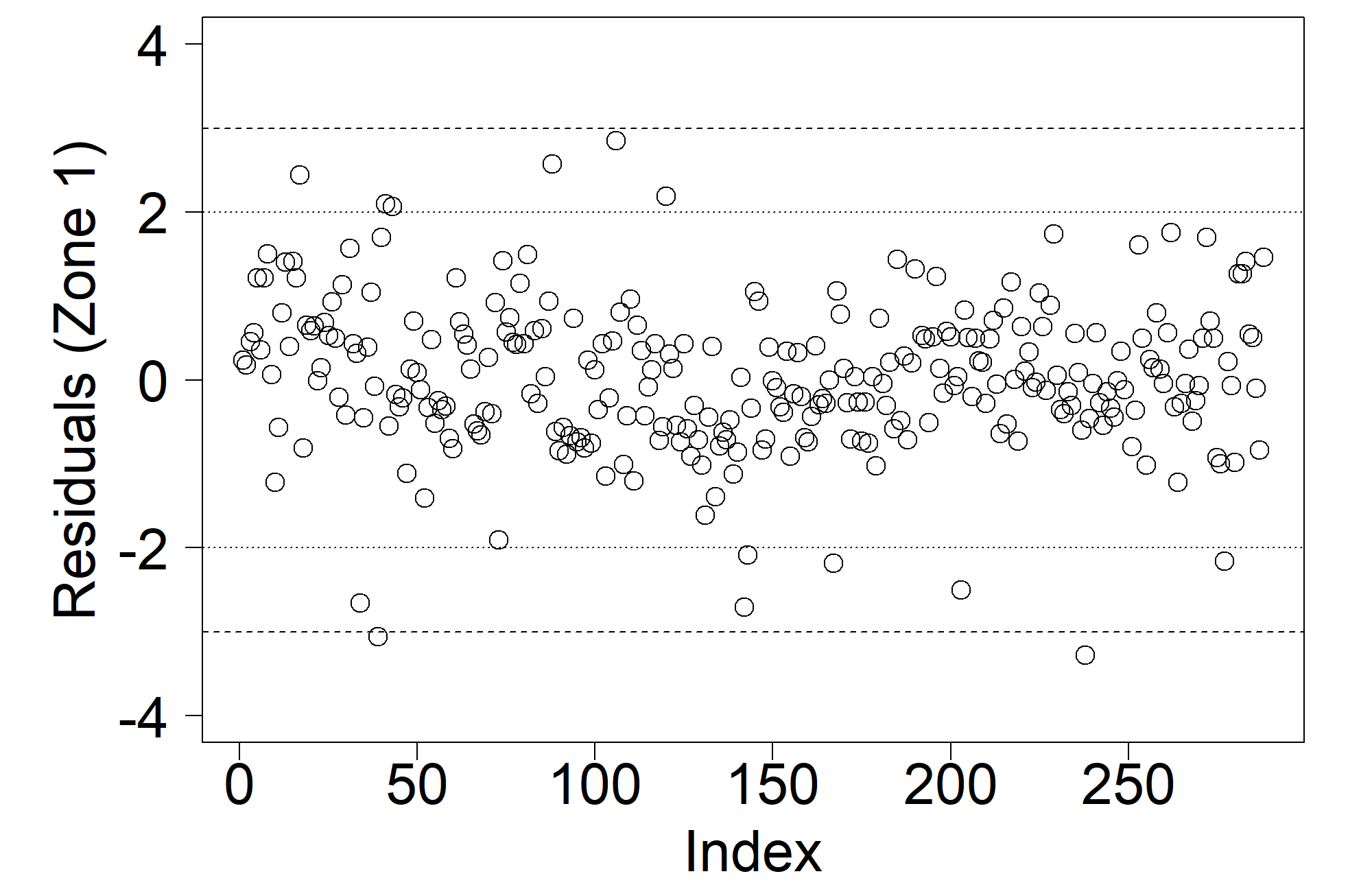}
		\caption*{(a)  }
	\end{minipage}\hfill
	\begin{minipage}[t]{0.48\textwidth}
		\centering
		\includegraphics[width=1\textwidth]{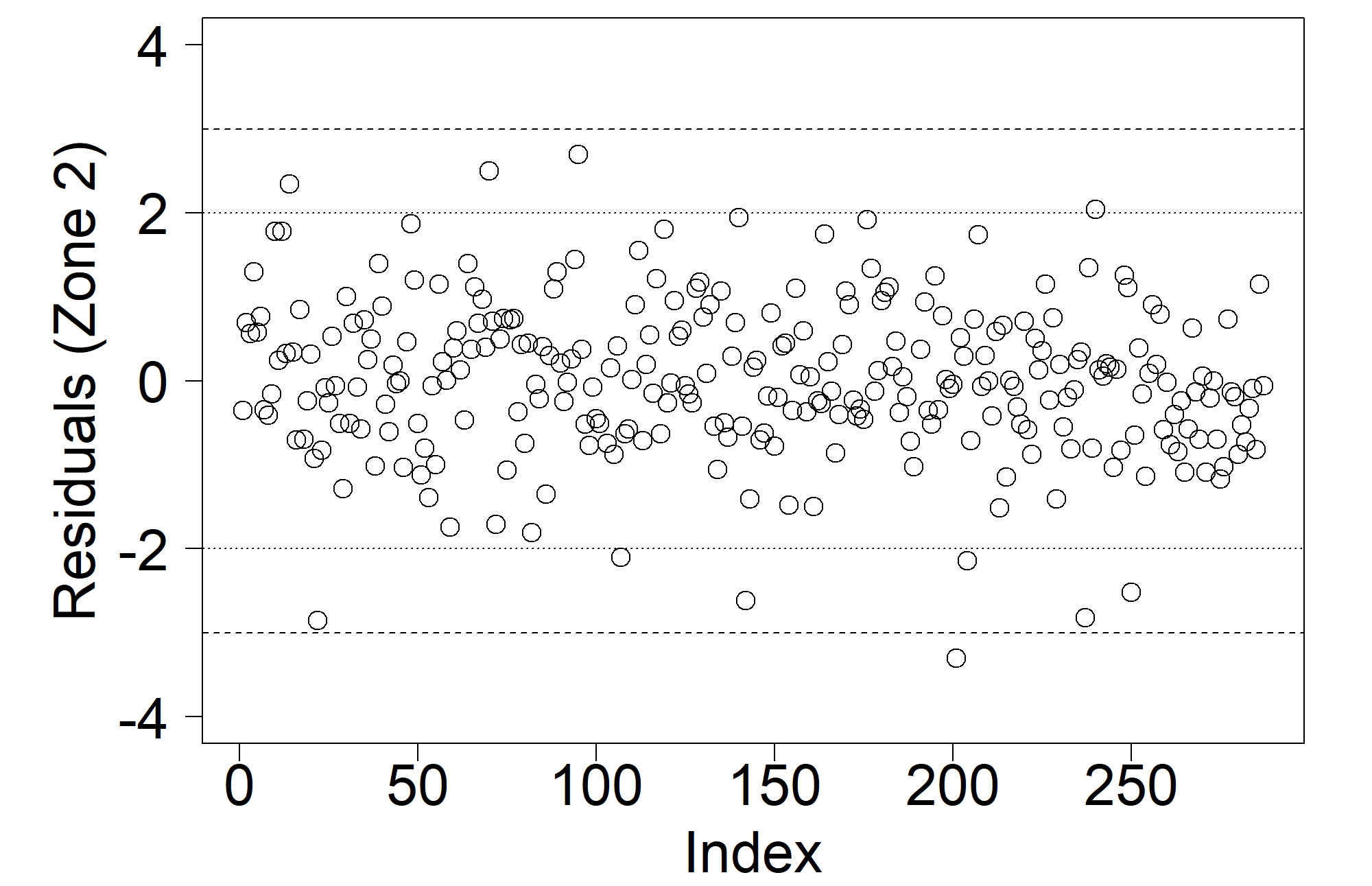}
		\caption*{(b) }
	\end{minipage}\hfill
	\begin{minipage}[t]{0.48\textwidth}
		\centering
		\includegraphics[width=1\textwidth]{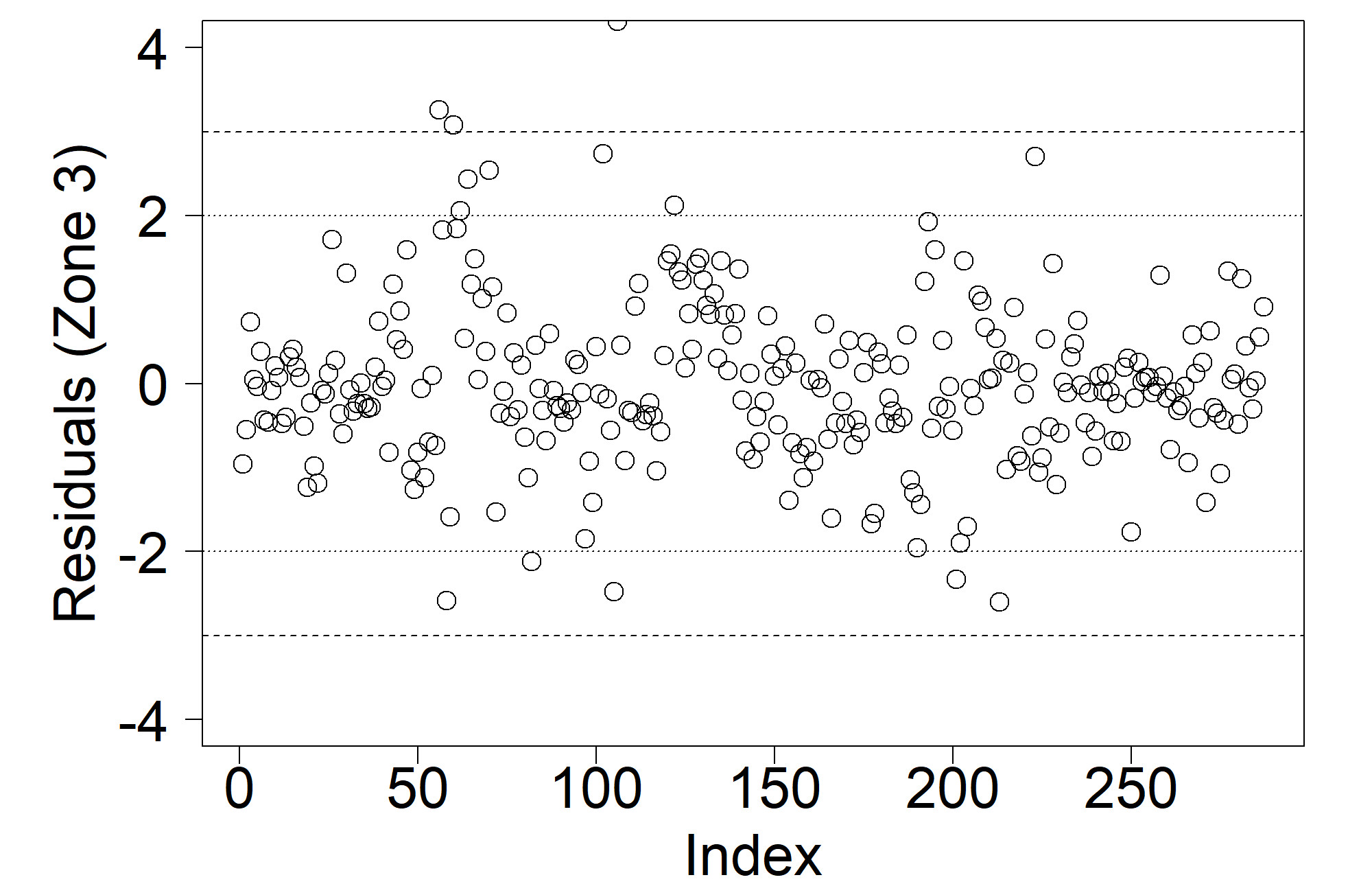}
		\caption*{(c) }
	\end{minipage}\hfill
	\begin{minipage}[t]{0.48\textwidth}
		\centering
		\includegraphics[width=1\textwidth]{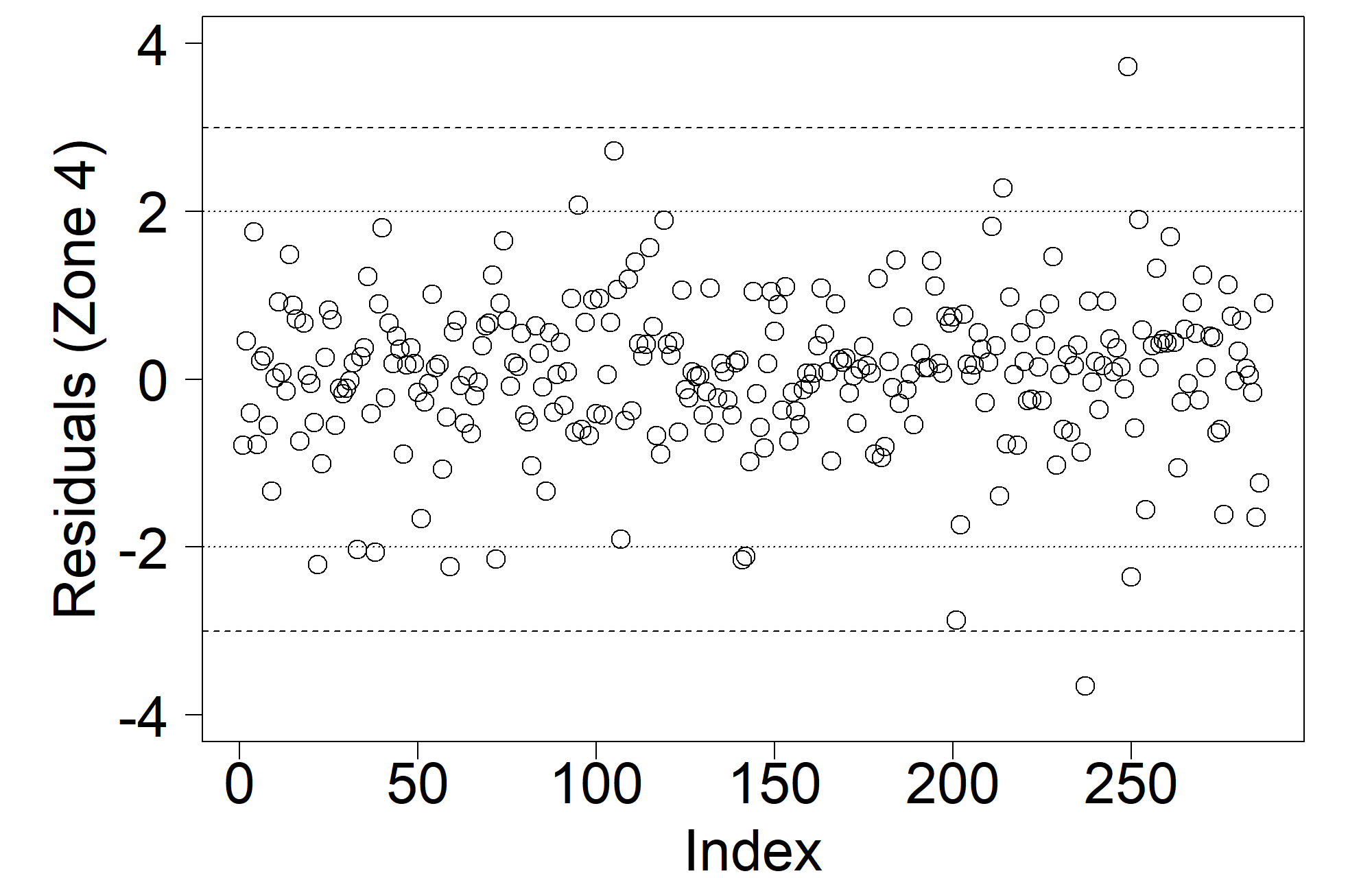}
		\caption*{(d) }
	\end{minipage}
	\caption{Residuals of fitted KTARMA model across SOM Zones.}
	\label{RES-KT}
\end{figure}

Figure~\ref{RES-KT} presents the residuals of the fitted KTARMA models for the four SOM-derived rainfall zones. The residuals fluctuate around zero without any discernible systematic pattern or persistent temporal structure, suggesting white-noise-like behaviour. This visual assessment is consistent with the Ljung--Box test results presented in Figure \ref{ljungtest}, which indicate no significant residual autocorrelation across the considered quantile levels.

\subsection{Forecast Comparison of KTARMA, KARMA and $\beta$ARMA}
To evaluate the predictive performance of the proposed frameworks, out-of-sample forecast accuracy is quantified across SOM rainfall zones. Table \ref{tablereg} provides a comparative summary of the MAE and MSE metrics computed over the testing period for the KTARMA, Kumaraswamy ARMA (KARMA) and \(\beta \)ARMA models and Figure \ref{Comparison} illustrates the 12-step out-of-sample forecast trajectories against actual observations of year 2025 across the four SOM zones. 

For Zone 1, all three modeling frameworks adequately reproduce the characteristic unimodal structure and high-volume nature of the rainfall cycle, with forecasts closely following the temporal evolution of the observed series. However, KTARMA provides the most accurate overall representation, attaining the lowest MAE and MSE. As illustrated in Figure \ref{Comparison}(a), although both KARMA and $\beta$ARMA capture the general seasonal evolution and timing of the monsoonal peak, KTARMA exhibits the closest agreement with the observed maximum rainfall, approaching the peak of approximately 800 mm at forecasting horizon (h=8).
\begin{table}[!ht]
	\centering 
	\caption{Forecast Accuracy Comparison of KTARMA, KARMA and $\beta$ARMA}
	\label{tablereg}
	\begin{tabular}{@{}lccccc@{}}
		\toprule
		\toprule
	 & Model & Zone 1 & Zone 2 & Zone 3 & Zone 4  \\ 
		\midrule
		\multirow{3}{*}{MAE} 
		& KTARMA      & 0.061037 & 0.052251 & 0.035931 & 0.091122 \\ 
		& KARMA       & 0.061156 & 0.052113 & 0.038773 & 0.084712 \\ 
		& $\beta$ARMA & 0.067843 & 0.053981 & 0.116424 & 0.088064 \\ 
		\midrule
		\multirow{3}{*}{MSE} 
		& KTARMA      & 0.007815 & 0.003695 & 0.001797 & 0.012386 \\ 
		& KARMA       & 0.012239 & 0.004209 & 0.002219 & 0.015758 \\ 
		& $\beta$ARMA & 0.009797 & 0.003799 & 0.014844 & 0.012550 \\ 
		\bottomrule
		\bottomrule
	\end{tabular}
\end{table}

For Zone 2, all three models demonstrate comparable forecasting performance in reproducing the bimodal seasonal rainfall pattern. According to Table \ref{tablereg}, KARMA achieves a marginally lower MAE (0.052113) than KTARMA (0.052251), indicating a slight advantage in terms of average absolute forecasting error. Nevertheless, KTARMA yields the lowest MSE, suggesting that it provides better control over relatively large deviations from the observed rainfall values. The forecasts shown in Figure \ref{Comparison}(b) further confirms all three models successfully reproduce the principal seasonal fluctuations and the overall bimodal structure.
\begin{figure}[ht]
	\begin{minipage}[t]{0.48\textwidth}
		\centering
		\includegraphics[width=1\textwidth]{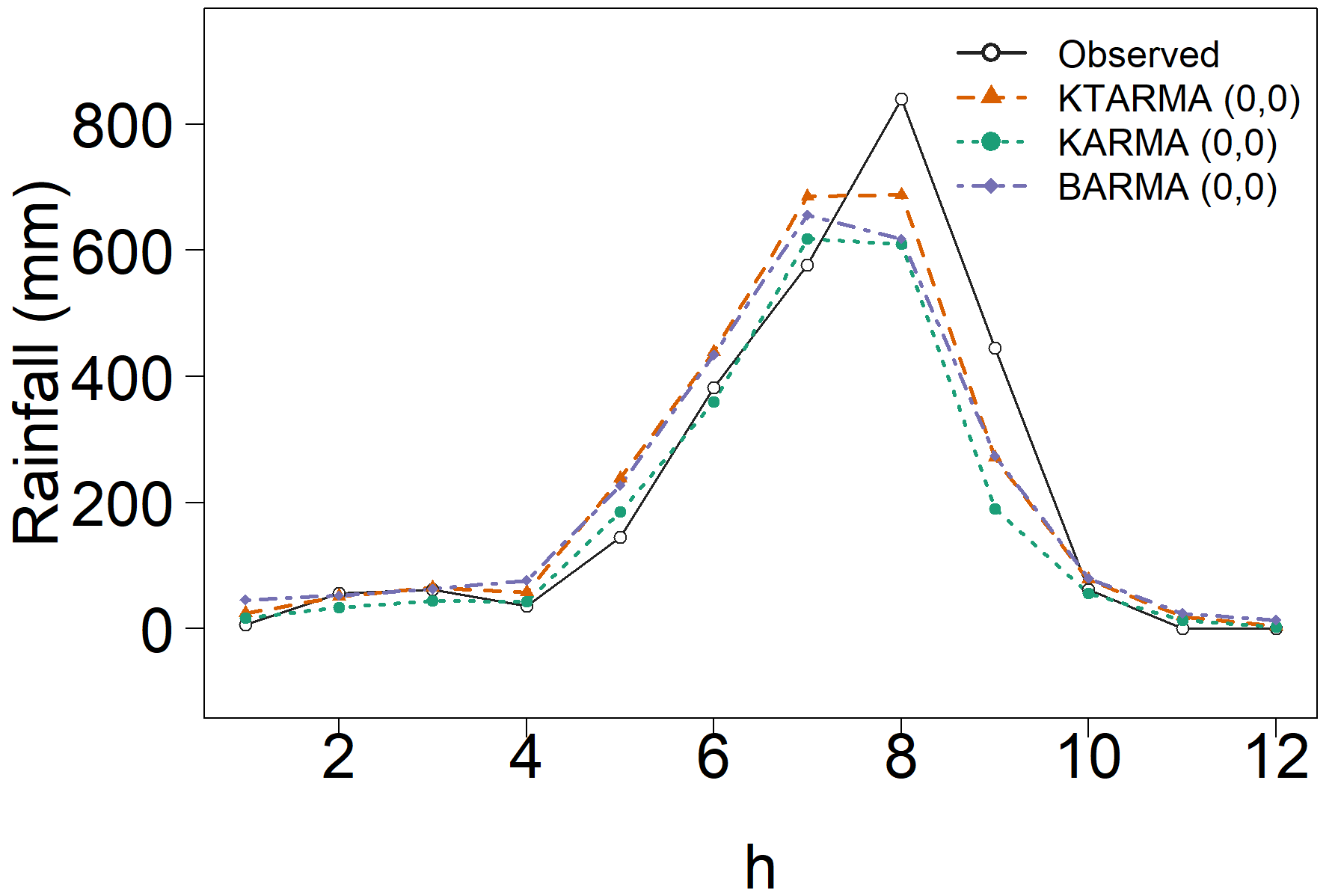}
		\caption*{(a) Zone 1 }
	\end{minipage}\hfill
	\begin{minipage}[t]{0.48\textwidth}
		\centering
		\includegraphics[width=1\textwidth]{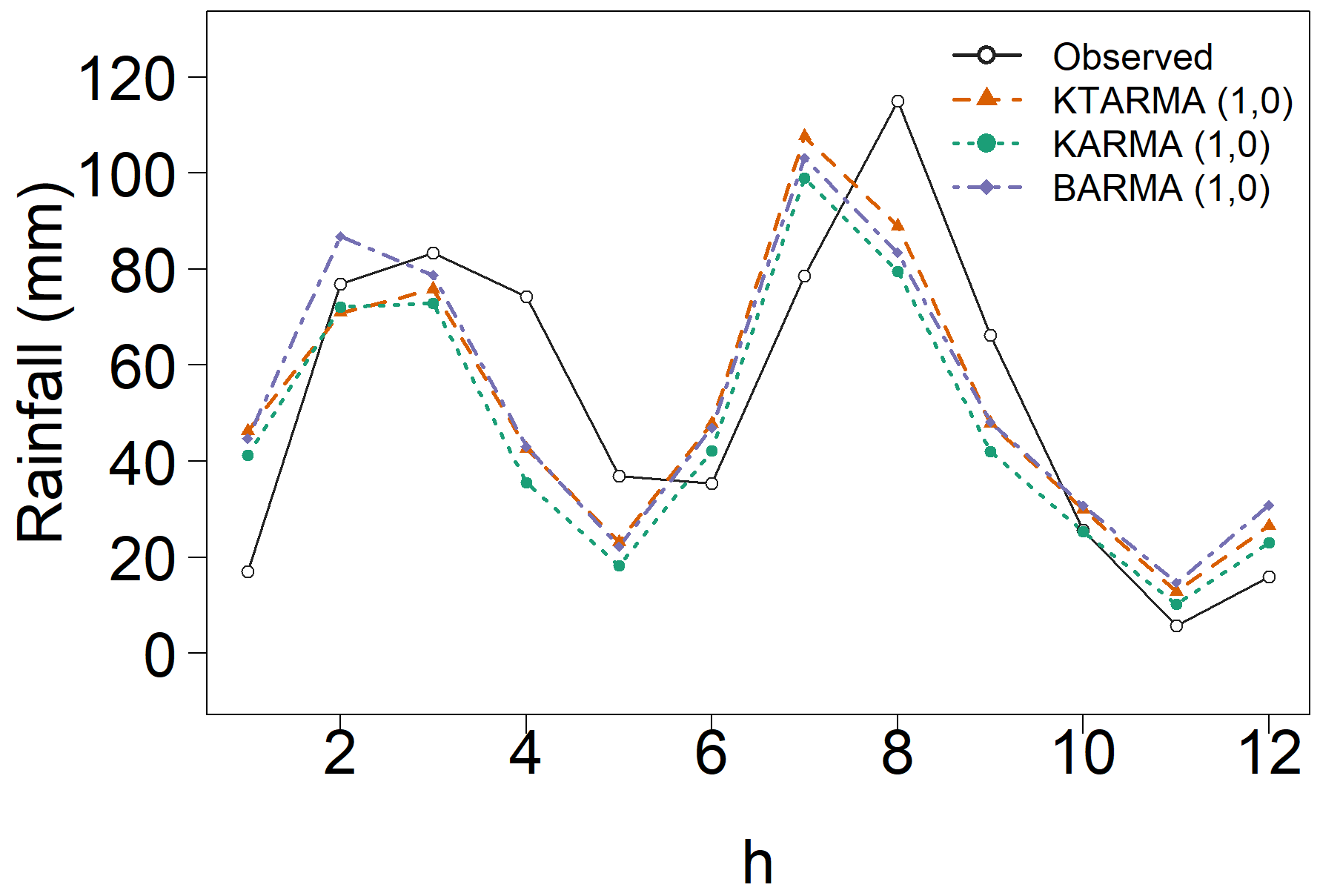}
		\caption*{(b) Zone 2}
	\end{minipage}\hfill
	\begin{minipage}[t]{0.48\textwidth}
		\centering
		\includegraphics[width=1\textwidth]{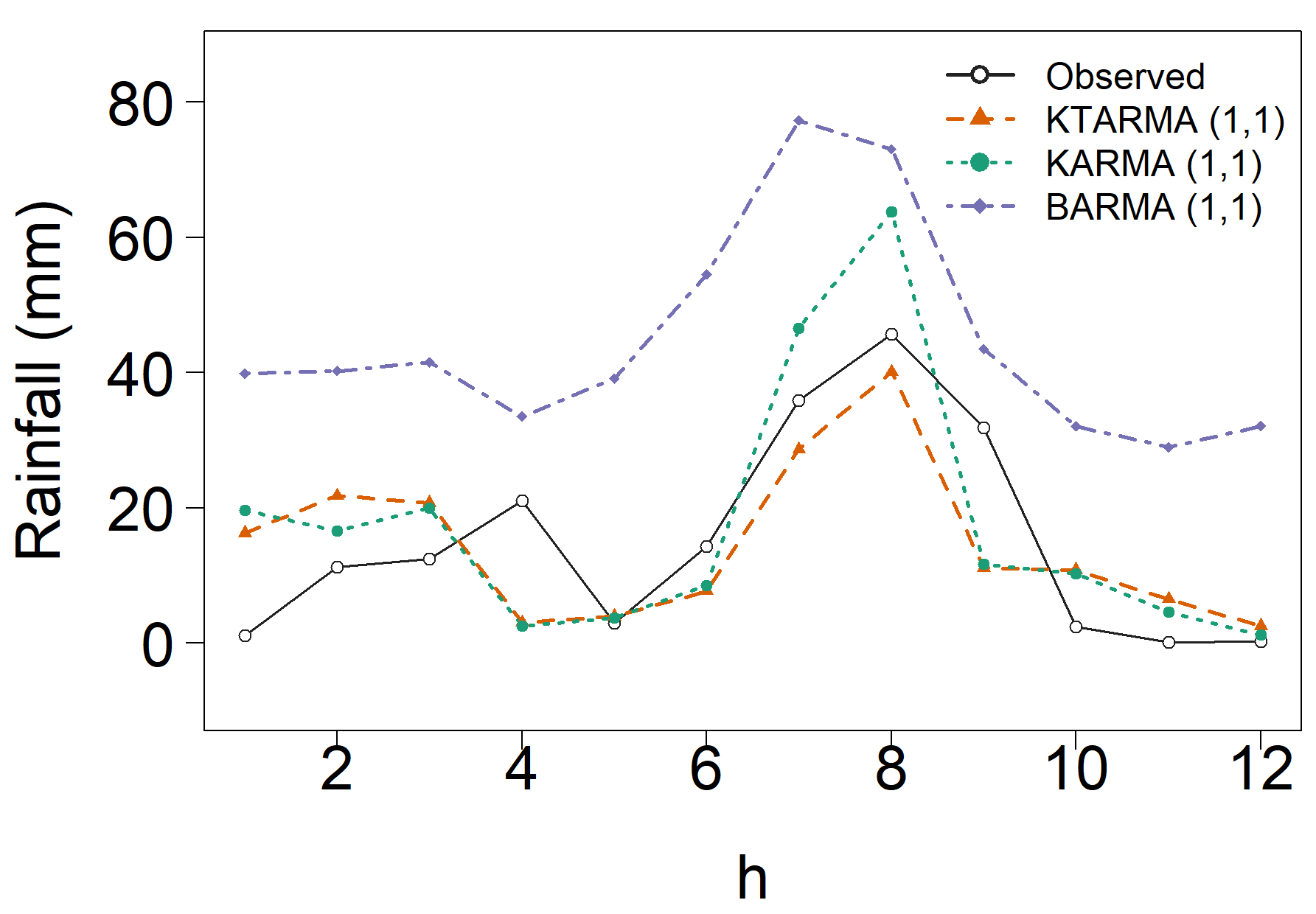}
		\caption*{(c) Zone 3}
	\end{minipage}\hfill
	\begin{minipage}[t]{0.48\textwidth}
		\centering
		\includegraphics[width=1\textwidth]{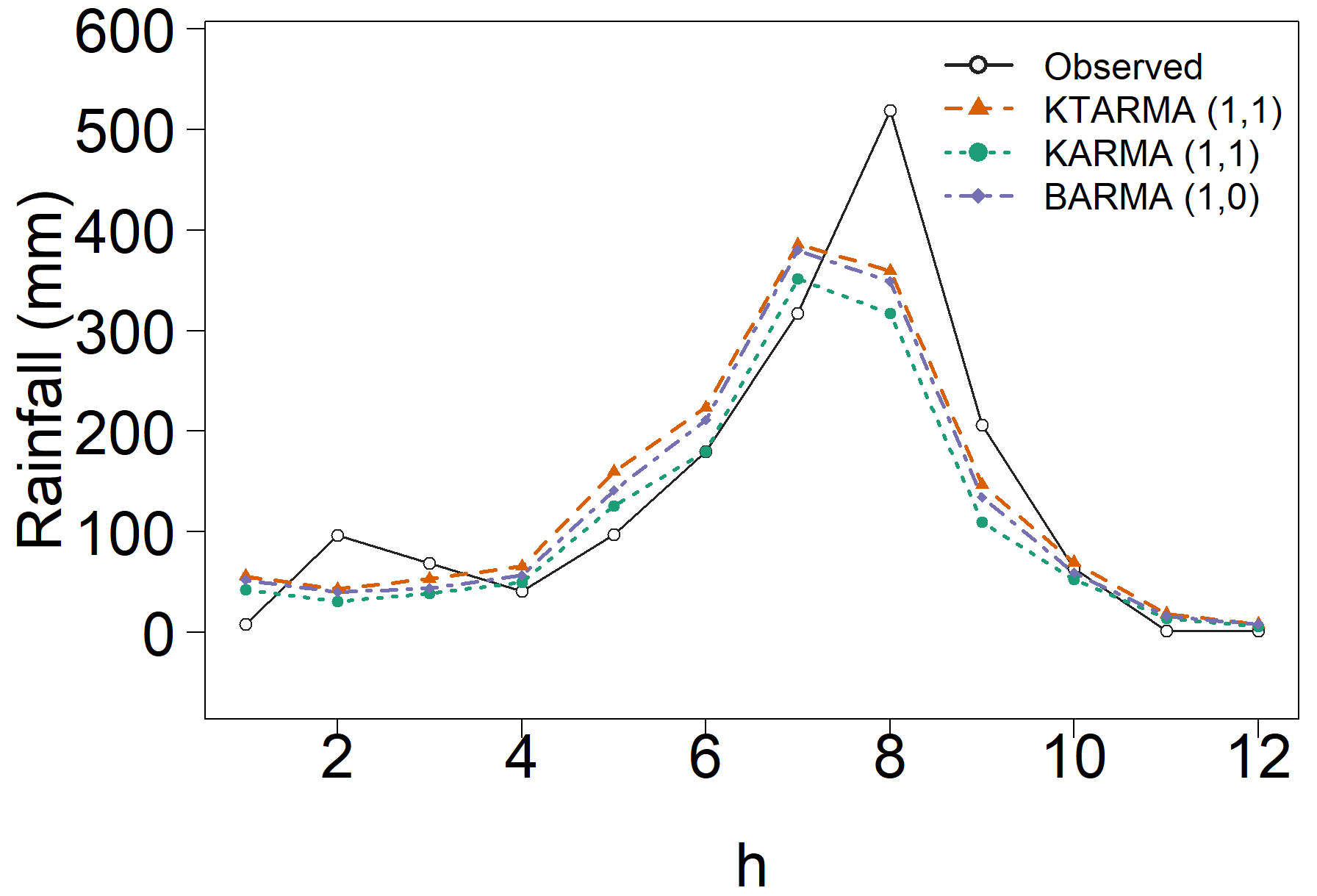}
		\caption*{(d) Zone 4}
	\end{minipage}
	\caption{Out-of-Sample Forecast Comparison of KTARMA, KARMA and $\beta$ARMA across SOM Zones.}
	\label{Comparison}
\end{figure}

For Zone 3, the comparative results demonstrate a clearer distinction among the modeling frameworks, particularly in the presence of a low-volume and highly skewed rainfall regime. As reported in Table \ref{tablereg}, KTARMA achieves lower MAE and MSE values than KARMA, indicating superior forecasting accuracy in this zone. In contrast, $\beta$ARMA exhibits substantially larger forecasting errors, suggesting a comparatively weaker ability to represent the underlying rainfall dynamics. The graphical comparison in Figure \ref{Comparison}(c) shows that KTARMA follows the observed rainfall trajectory more closely, particularly across periods of peak and declining rainfall, whereas KARMA tends to overestimate the higher rainfall values. The $\beta$ARMA forecasts exhibit a persistent tendency toward overestimation over much of the forecasting period, resulting in a comparatively poor representation of the observed seasonal variability. 

For Zone 4, the results indicate a complementary performance between KARMA and KTARMA. KARMA attains the lowest MAE (0.084712), indicating the smallest average absolute deviation from the observed series, whereas KTARMA achieves the lowest MSE (0.012386), demonstrating superior control of larger forecasting errors. The graphical results presented in Figure \ref{Comparison}(d) show that all three models successfully reproduce the major seasonal evolution, including the peak rainfall.

\section{Conclusion} \label{conclusion}
In this work, we introduce the KTARMA model for positive-valued time series by combining the Kumaraswamy--Teissier distribution with an ARMA structure and a general conditional quantile formulation. A conditional maximum likelihood framework is developed, and explicit expressions for the conditional score vector and information matrix are derived. The finite-sample performance of the proposed estimators is investigated through Monte Carlo simulations under two KTARMA configurations and at multiple quantile levels. The results show that the parameter estimates become more accurate as the sample size increases, with decreasing bias and mean squared error. The proposed framework is applied to monthly rainfall data from the NWH. The KTARMA model accommodates serial dependence, seasonal variation, and external covariates while allowing the conditional distribution to be examined at different quantile levels. The empirical results indicate substantial variation in temporal dependence and covariate effects across rainfall zones. Diagnostic results show no significant remaining serial correlation, and the forecasting analysis indicates that KTARMA generally outperforms KARMA and $\beta$ARMA, particularly in terms of MSE, while maintaining competitive MAE values. Overall, KTARMA provides a flexible observation-driven framework for modelling positive-valued time series with asymmetric distributions, temporal dependence, seasonal effects, and covariates. The conditional quantile formulation further allows the dynamics of different parts of the conditional distribution to be captured within a unified modelling framework.

\subsection*{Acknowledgments}
We express our gratitude to all the data providers for their contributions to the statistical data analysis, duly acknowledged through appropriate citations.

\appendix

\section{Appendix A} \label{app:derivation}
In this appendix, we present the explicit second-order partial derivatives of the log-likelihood function obtained via CMLE for the proposed KTARMA model.

\begin{equation} \label{rho2}
	\begin{aligned} 
		\frac{\partial^2 \ell_t}{\partial \mu_{\rho,t}^2} &= a \lambda \left[ \frac{\left( M_t^{a-1} e^{\psi(\mu_{\rho,t};\lambda)} \lambda \left[ (a-1) \frac{e^{\psi(\mu_{\rho,t};\lambda)}  (e^{\lambda \mu_{\rho,t}} - 1)^2}{M_t} + (e^{\lambda \mu_{\rho,t}} - 1) + 1 - (e^{\lambda \mu_{\rho,t}} - 1)^2 \right] \right) }{(1 - M_t^a) \log(1-M_t^a)} \right. \\
		&\quad \left. - \frac{M_t^{a-1} (e^{\lambda \mu_{\rho,t}} - 1) e^{\psi(\mu_{\rho,t};\lambda)}  \cdot \left( -a M_t^{a-1} e^{\psi(\mu_{\rho,t};\lambda)}  \lambda (e^{\lambda \mu_{\rho,t}} - 1) (\log(1-M_t^a)+1) \right)}{(1 - M_t^a)^2 \log^2(1-M_t^a)} \right] \cdot H \\
		&\quad + \frac{a \lambda M_t^{a-1} (e^{\lambda \mu_{\rho,t}} - 1) e^{\psi(\mu_{\rho,t};\lambda)} }{(1 - M_t^a) \log(1-M_t^a)} \cdot \frac{c_t \log(1-A_t^a) \cdot a M_t^{a-1} e^{\psi(\mu_{\rho,t};\lambda)}  \lambda (e^{\lambda \mu_{\rho,t}} - 1)}{(1 - M_t^a) \log(1-M_t^a)}
	\end{aligned}
\end{equation}
where, $H = \left(1 + c_t \log(1-A_t^a)\right)$

\begin{equation} \label{rhoa}
	\begin{aligned}
		\frac{\partial^2 l_t}{\partial a \partial \mu_{\rho,t}}  &= \frac{\lambda (e^{\lambda \mu_{\rho,t}} - 1) e^{\fmt} }{M_t} \left[ \frac{M_t^a \cdot H}{(1 - M_t^a) \log(1 - M_t^a)} \right. + a \cdot \frac{M_t^a \log(M_t) \left( \log(1 - M_t^a) + M_t^a \right)}{(1 - M_t^a)^2 \left[\log(1 - M_t^a)\right]^2} H \\
		&\quad \left. +  \frac{a c_t M_t^a}{(1 - M_t^a) } \cdot \frac{\frac{-A_t^a \log(A_t)}{1 - A_t^a} \log(1 - M_t^a) - \log(1 - A_t^a) \cdot \frac{-M_t^a \log(M_t)}{1 - M_t^a}}{\left[\log(1 - M_t^a)\right]^2} \right]
	\end{aligned}
\end{equation}

\begin{equation}\label{rholambda}
	\begin{aligned} &\frac{\partial^2 l_t}{\partial \lambda \partial \mu_{\rho,t}}  = a \mu_{\rho,t} \Bigg[ - \frac{ M_t^{a-1} (e^{\lambda \mu_{\rho,t}} - 1) e^{\fmt} \cdot \left( -a M_t^{a-1} e^{\fmt} \mu_{\rho,t} (e^{\lambda \mu_{\rho,t}} - 1) \big(\log(1 - M_t^a) + 1\big) \right) }{(1 - M_t^a)^2 (\log(1 - M_t^a))^2} \cdot H  \\ &\frac{ \left( M_t^{a-1} e^{\fmt} \mu_{\rho,t} \left[ (a-1)\frac{e^{\fmt} (e^{\lambda \mu_{\rho,t}} - 1)^2}{M_t} + (e^{\lambda \mu_{\rho,t}} - 1) + 1 - (e^{\lambda \mu_{\rho,t}} - 1)^2 \right] \right)  }{(1 - M_t^a) (\log(1 - M_t^a))} \cdot H  \\ &+ \frac{c_t M_t^{a-1} (e^{\lambda \mu_{\rho,t}} - 1) e^{\fmt}}{1 - M_t^a} \Bigg[ \frac{ \left( \frac{-a A_t^{a-1} e^{\psi(y_t; \lambda)} y_t (e^{\lambda y_t} - 1)}{1 - A_t^a} \right) }{\log(1 - M_t^a)} - \frac{  \log(1 - A_t^a) \left( \frac{-a M_t^{a-1} e^{\fmt} \mu_{\rho,t} (e^{\lambda \mu_{\rho,t}} - 1)}{1 - M_t^a} \right) }{(\log(1 - M_t^a))^2}\Bigg] \Bigg] \end{aligned}
\end{equation}

\begin{equation}\label{a2}
\begin{aligned}
	\frac{\partial^2 l_t}{\partial a^2} &= -\frac{1}{a^2} + \frac{M_t^a (\log M_t)^2 \left(\log(1 - M_t^a) + M_t^a\right)}{(1 - M_t^a)^2 \left[\log(1 - M_t^a)\right]^2} \cdot H - \left[ \left( c_t - 1 \right) \cdot \frac{A_t^a (\log A_t)^2}{(1 - A_t^a)^2} \right]\\
	&\quad + \frac{c_t M_t^a \log M_t }{(1 - M_t^a) \left[\log(1 - M_t^a)\right]^2} \cdot \left[\frac{- 2 A_t^a \log A_t}{1 - A_t^a} \log(1 - M_t^a) - \log(1 - A_t^a) \frac{-M_t^a \log M_t}{1 - M_t^a} \right]
\end{aligned}
\end{equation}

\begin{equation} \label{alambda}
	\begin{aligned} \frac{\partial^2 l_t}{\partial \lambda \partial a } &= \frac{y_t (e^{\lambda y_t} - 1)e^{\fyt} }{A_t} + \left[ \frac{c_t M_t^a \log M_t}{(1 - M_t^a)\log(1 - M_t^a)} \right] \\ &\quad \times \frac{\left( \frac{-a A_t^{a-1} y_t (e^{\lambda y_t} - 1) e^{\fyt} }{1 - A_t^a} \right) \log(1 - M_t^a) + \log(1 - A_t^a) \left( \frac{a M_t^{a-1} \mu_{\rho,t} (e^{\lambda \mu_{\rho,t}} - 1) e^{\fmt}}{1 - M_t^a} \right)}{\log(1 - M_t^a)} \\ &\quad + \left[ \frac{\left( a M_t^{a-1} \mu_{\rho,t} (e^{\lambda \mu_{\rho,t}} - 1) e^{\fmt} \cdot \log M_t + M_t^a \cdot \frac{\mu_{\rho,t}(e^{\lambda \mu_{\rho,t}}-1)e^{\fmt}}{M_t} \right)(1 - M_t^a)\log(1 - M_t^a)}{(1 - M_t^a)^2 [\log(1 - M_t^a)]^2} \right] \cdot H \\  &\quad - \frac{c_t \cdot a M_t^{a-1} \mu_{\rho,t}(e^{\lambda \mu_{\rho,t}} - 1)e^{\fmt}}{(1 - M_t^a)\log(1 - M_t^a)} \cdot \frac{A_t^a \log A_t}{1 - A_t^a}  - (c_t - 1) \cdot \frac{a A_t^{a-1} y_t (e^{\lambda y_t} - 1) e^{\fyt} (1 - A_t^a + \log A_t)}{(1 - A_t^a)^2} \end{aligned}
\end{equation}

\begin{align} \label{lam2}
	&\frac{\partial^2 \ell_t}{\partial \lambda^2}
	= -\frac{1}{\lambda^2}
	- \frac{y_t^2 e^{\lambda y_t}}{(e^{\lambda y_t} - 1)^2}
	- y_t^2 e^{\lambda y_t} \nonumber \\
	&\quad + \frac{y_t^2 e^{\psi(y_t;\lambda)}}{A_t^2}
	\Big[ \big(e^{\lambda y_t} - (e^{\lambda y_t} - 1)^2\big) A_t - e^{\psi(y_t;\lambda)} (e^{\lambda y_t} - 1)^2 \Big]
	\left[ (a-1) - \frac{a (c_t-1) A_t^a}{1 - A_t^a} \right] \nonumber \\
	&\quad - \frac{a y_t e^{\psi(y_t;\lambda)}(e^{\lambda y_t}-1)}{A_t}
	\Bigg[
	\frac{ a \mu_t c_t M_t^{a-1}(1-M_t)(e^{\lambda \mu_t}-1) A_t^a}{(1-M_t^a)\log(1-M_t^a)(1-A_t^a)}
	+ \frac{ a (c_t-1) y_t A_t^{a-1}(1-A_t)(e^{\lambda y_t}-1)}{(1-A_t^a)^2}
	\Big] \nonumber \\
	&\quad + \frac{a \mu_t}{(1-M_t^a)^2 [\log(1-M_t^a)]^2}
	\Bigg[
	\Bigg\{
	M_t^{a-1} e^{\psi(\mu_t;\lambda)} \mu_t \left[ (a-1) \frac{\mu_t (1-M_t)(e^{\lambda \mu_t}-1)}{M_t} + \frac{\mu_t e^{\lambda \mu_t}}{e^{\lambda \mu_t}-1} - \mu_t(e^{\lambda \mu_t}-1) \right]  \nonumber \\
	&\qquad \times (1-M_t^a)\log(1-M_t^a) + a M_t^{2a-2} e^{2\psi(\mu_t;\lambda)} \mu_t (1-M_t)(e^{\lambda \mu_t}-1)^2 \left( \log(1-M_t^a) + 1 \right)
	\Bigg\} \cdot H
	\Bigg] \nonumber \\
	&\quad + \frac{ a \mu_t M_t^{a-1} (e^{\lambda \mu_t} - 1) e^{\psi(\mu_t;\lambda)} }{(1-M_t^a) \log(1-M_t^a)}
	\Bigg[
	\frac{ a \mu_t c_t M_t^{a-1}(1-M_t)(e^{\lambda \mu_t}-1) \log(1 - A_t^a)}{(1-M_t^a) \log(1-M_t^a)}
	- \frac{ a y_t c_t A_t^{a-1}(1-A_t)(e^{\lambda y_t}-1)}{1-A_t^a}
	\Bigg] 
\end{align}

\section{Appendix B} \label{app:lemma}
In this appendix, we present the results required for obtaining the conditional information matrix of the introduced KTARMA model.

\begin{lemma}
	Let $Y_t$ be a random variable whose conditional distribution given
	$\mathcal{F}_{t-1}$ is specified by KT$(\mu_{\rho,t}, a, \lambda)$. Then
	\[
		\mathbb{E}(\log(1-A_t^a)|\mathcal{F}_{t-1}) = -\frac{1}{c_t}
	\]	\label{lem:lemma1}
\end{lemma}
\begin{proof}
\begin{align*}
	\mathbb{E}(\log(1-A_t^\a) \mid \mathcal{F}_{t-1}) &= \int_0^\infty \log(1-A_t^\a) f(y) \, dy \\
	&= \int_0^\infty \log(1-A_t^\a) \a \lambda c_t (e^{\lambda y} -1) e^{\psi(y;\lambda)} \left(1-e^{\psi(y;\lambda)} \right)^{\a-1} (1-A_t^\a)^{c_t-1} \, dy
\end{align*}
By substituting, 
\[u = \log(1 - A_t^\a) = \log(1 - (1 - e^{\psi(y;\lambda)})^\a)\]
\[du = \frac{-\a \lambda (e^{\lambda y} - 1) e^{\psi(y;\lambda)} (1 - e^{\psi(y;\lambda)})^{\a-1}}{1 - A_t^\a} \, dy\]
Transforming the integration limits adjusts the boundaries from $y \in [0, \infty)$ to $u \in [0, -\infty)$ and hence,
\begin{align*}
	\mathbb{E}(\log(1-A_t^\a) \mid \mathcal{F}_{t-1}) &= -c_t \int_0^{-\infty} u \cdot (e^u)^{c_t} \, du = c_t \int_{-\infty}^0 u e^{c_t u} \, du \\ &= c_t \left( -\frac{1}{c_t^2} \right) = -\frac{1}{c_t} 
\end{align*}
\end{proof}

\begin{lemma}
	Let $Y_t$ be a random variable whose conditional distribution given
$\mathcal{F}_{t-1}$ is specified by KT$(\mu_{\rho,t}, a, \lambda)$. Then
\[
\mathbb{E}\left(\frac{A_t^a \log(A_t)}{1-A_t^a} \;\middle\vert{}\; \mathcal{F}_{t-1}\right) = \frac{\psi_0(2) - \psi_0(c_t+1)}{a(c_t-1)}
\]
\label{lem:lemma2}
\end{lemma}
\begin{proof}
\begin{align*}
\mathbb{E}\left(\frac{A_t^a \log(A_t)}{1-A_t^a} \;\middle\vert{}\; \mathcal{F}_{t-1}\right) 
	&= \int_0^\infty \frac{A_t^a \log(A_t)}{1-A_t^a} \a \lambda c_t (e^{\lambda y} -1) e^{\psi(y;\lambda)} \left(1-e^{\psi(y;\lambda)} \right)^{\a-1} (1-A_t^\a)^{c_t-1} \, dy
\end{align*}
By substituting, 
\[u = A_t = 1 - e^{\psi(y;\lambda)}\] 
\[du = \lambda (e^{\lambda y} - 1) e^{\psi(y;\lambda)} dy \]
Transforming the integration domain boundaries to $u \in [0, 1]$. Therefore, 
\begin{equation*}
\mathbb{E}\left(\frac{A_t^a \log(A_t)}{1-A_t^a} \;\middle\vert{}\; \mathcal{F}_{t-1}\right) =	a c_t \int_0^1 \log(u) \, u^{2a-1} (1-u^a)^{c_t-2} \, du
\end{equation*}
Further, applying transformation $v = u^a$, the integration becomes:
\begin{equation*}
\mathbb{E}\left(\frac{A_t^a \log(A_t)}{1-A_t^a} \;\middle\vert{}\; \mathcal{F}_{t-1}\right) = 	\frac{c_t}{a} \int_0^1 v \log(v) (1-v)^{c_t-2} \, dv
\end{equation*}
By setting \(x = 2\) and \(y = c_t - 1\), we have 
\[\int _{0}^{1}v\log (v)(1-v)^{c_{t}-2}\,dv=\left.\frac{\partial }{\partial x}B(x,c_{t}-1)\right|{}_{x=2}\]
Using the Digamma identity \(\frac{\partial}{\partial x} B(x,y) = B(x,y)[\psi_0(x) - \psi_0(x+y)]\), where \(\psi_0(\cdot)\) is the Digamma function, we can evaluate:
\[B(2,c_{t}-1)=\frac{\Gamma (2)\Gamma (c_{t}-1)}{\Gamma (2+c_{t}-1)}=\frac{1\cdot \Gamma (c_{t}-1)}{c_{t}(c_{t}-1)\Gamma (c_{t}-1)}=\frac{1}{c_{t}(c_{t}-1)}\]
Hence, the solution is given by: 
\begin{equation*}
	\mathbb{E}\left(\frac{A_t^a \log(A_t)}{1-A_t^a} \;\middle\vert{}\; \mathcal{F}_{t-1}\right) = \frac{\psi_0(2) - \psi_0(c_t+1)}{a(c_t-1)}
\end{equation*}
\end{proof}

\begin{lemma}
	Let $Y_t$ be a random variable whose conditional distribution given
	$\mathcal{F}_{t-1}$ is specified by KT$(\mu_{\rho,t}, a, \lambda)$. Then
		\[\mathbb{E}\left(\frac{A_t^a ( \log(A_t))^2}{(1-A_t^a)^2} \;\middle\vert{}\; \mathcal{F}_{t-1}\right) = \frac{c_t}{a^2(c_t-1)(c_t-2)} \left[ \big(\psi_0(2) - \psi_0(c_t)\big)^2 + \psi_1(2) - \psi_1(c_t) \right] \]
	\label{lem:lemma3}
\end{lemma}
\begin{proof}
	\begin{align*}
		\mathbb{E}\left(\frac{A_t^a ( \log(A_t))^2}{(1-A_t^a)^2} \;\middle\vert{}\; \mathcal{F}_{t-1}\right)
		&= \int_0^\infty \frac{A_t^a ( \log(A_t))^2}{(1-A_t^a)^2} \a \lambda c_t (e^{\lambda y} -1) e^{\psi(y;\lambda)} \left(1-e^{\psi(y;\lambda)} \right)^{\a-1} (1-A_t^\a)^{c_t-1} \, dy
	\end{align*}
	By performing similar transformation as done in Lemma \ref{lem:lemma2}, we have:
	\begin{align*}
\mathbb{E}\left(\frac{A_t^a ( \log(A_t))^2}{(1-A_t^a)^2} \;\middle\vert{}\; \mathcal{F}_{t-1}\right)
	&=	\frac{c_t}{a^2} \int_0^1 v (\log v)^2 (1-v)^{c_t-3} \, dv
	\end{align*}
By setting \(x = 2\) and \(y = c_t - 2\),
\[\int _{0}^{1}v(\log v)^{2}(1-v)^{c_{t}-3}\,dv=\left.\frac{\partial ^{2}}{\partial x^{2}}B(x,c_{t}-2)\right|_{x=2}\]
Since,\(\frac{\partial ^{2}}{\partial x^{2}}B(x,y)=B(x,y)\left[\big(\psi _{0}(x)-\psi _{0}(x+y)\big)^{2}+\psi _{1}(x)-\psi _{1}(x+y)\right]\)where \(\psi_0(\cdot)\) is the Digamma function and \(\psi_1(\cdot)\) represents the Trigamma function, we can evaluate:
 \[B(2,c_{t}-2)=\frac{\Gamma (2)\Gamma (c_{t}-2)}{\Gamma (c_{t})}=\frac{1\cdot \Gamma (c_{t}-2)}{(c_{t}-1)(c_{t}-2)\Gamma (c_{t}-2)}=\frac{1}{(c_{t}-1)(c_{t}-2)}\]
	Hence, the final expression is: 
	\[\mathbb{E}\left(\frac{A_t^a ( \log(A_t))^2}{(1-A_t^a)^2} \;\middle\vert{}\; \mathcal{F}_{t-1}\right) = \frac{c_t}{a^2(c_t-1)(c_t-2)} \left[ \big(\psi_0(2) - \psi_0(c_t)\big)^2 + \psi_1(2) - \psi_1(c_t) \right] \]
	where,  \(\psi_0(2) = 1 - \kappa_E\) and \(\psi_1(2) = \frac{\pi^2}{6} - 1\), where \(\kappa_E \approx 0.5772156649...\) is the Euler-Mascheroni constant (\cite{gradshteyn2014table}).
\end{proof}

\begin{lemma}
	Let $Y_t$ be a random variable whose conditional distribution given
	$\mathcal{F}_{t-1}$ is specified by KT$(\mu_{\rho,t}, a, \lambda)$. Then
	\[\mathbb{E}\left(\frac{A_t^{a-1} y_t e^{\psi(y_t;\lambda)} (e^{\lambda y_t} -1) }{(1-A_t^a)} \;\middle|\; \mathcal{F}_{t-1}\right) = \frac{a c_t}{\lambda} \sum_{k=0}^{\infty}\sum_{j=0}^{\infty} (-1)^{k+j} e^{j+2} \binom{c_t-2}{k}\binom{a(k+2)-2}{j} I_{\text{comp}}(j) \]
	\label{lem:lemma4}
\end{lemma}

\begin{proof}
	\begin{align*}
		&\mathbb{E}\left(\frac{A_t^{a-1} y_t e^{\psi(y_t;\lambda)} (e^{\lambda y_t} -1) }{(1-A_t^a)} \;\middle\vert{}\; \mathcal{F}_{t-1}\right)\\
		&= \int_0^\infty \frac{A_t^{a-1} y_t e^{\psi(y_t;\lambda)} (e^{\lambda y_t} -1) }{(1-A_t^a)} \a \lambda c_t (e^{\lambda y} -1) e^{\psi(y;\lambda)} \left(1-e^{\psi(y;\lambda)} \right)^{\a-1} (1-A_t^\a)^{c_t-1} \, dy\\
		&= 	a \lambda c_t \int_0^\infty y (e^{\lambda y} - 1)^2 e^{2\psi(y;\lambda)} A_t^{2a-2} (1 - A_t^a)^{c_t-2} \, dy
	\end{align*}
	By using generalized binomial series expansion for $(1-A_t^a)^{c_t-2}$,
	\[\mathbb{E}=a\lambda c_{t}\sum _{k=0}^{\infty }(-1)^{k}{c_{t}-2 \choose k}\int _{0}^{\infty }y(e^{\lambda y}-1)^{2}e^{2\psi (y;\lambda )}A_{t}^{a(k+2)-2}\,dy\]
	Substituting $A_t = 1 - e^{\psi(y;\lambda)}$, a secondary binomial expansion results the integral into a double infinite series:
	\[a \lambda c_t \sum_{k=0}^{\infty} \sum_{j=0}^{\infty} (-1)^{k+j} \binom{c_t - 2}{k} \binom{a(k+2)-2}{j} e^{j+2} \int_0^\infty y \left[ e^{(j+4)\lambda y} - 2e^{(j+3)\lambda y} + e^{(j+2)\lambda y} \right] e^{-(j+2)e^{\lambda y}} \, dy\]
	Applying the transformation $z = e^{\lambda y}$, which implies $dy = \frac{1}{\lambda z} dz$, maps the integration boundaries to $z \in [1, \infty)$ simplifies the expression as:
	\[	\frac{a c_t e^{j+2}}{\lambda} \sum_{k=0}^{\infty} \sum_{j=0}^{\infty} (-1)^{k+j} \binom{c_t - 2}{k} \binom{a(k+2)-2}{j} \int_1^\infty \ln(z) \left[ z^{j+3} - 2z^{j+2} + z^{j+1} \right] e^{-(j+2)z} \, dz\]
	These internal integral components can be solved analytically using the properties of the upper incomplete gamma function. Let $b = j+2$, $\Gamma(m, b) = \int_b^\infty t^{m-1}e^{-t}dt$, and let $\Gamma'(m, b) = \frac{\partial}{\partial m}\Gamma(m, b)$ denote its derivative. Then,
	\begin{align*}
		I_{\text{comp}}(j) = &\frac{1}{b^{j+4}} \left[ \Gamma'(j+4, b) - \ln(b) \Gamma(j+4, b) \right] - \frac{2}{b^{j+3}} \left[ \Gamma'(j+3, b) - \ln(b) \Gamma(j+3, b) \right] \\
		&+ \frac{1}{b^{j+2}} \left[ \Gamma'(j+2, b) - \ln(b) \Gamma(j+2, b) \right]
	\end{align*}
	Therefore, the final solution is given by:
\[	\mathbb{E}\left(\frac{A_t^{a-1} y_t e^{\psi(y_t;\lambda)} (e^{\lambda y_t} -1) }{(1-A_t^a)} \;\middle|\; \mathcal{F}_{t-1}\right) = \frac{a c_t}{\lambda} \sum_{k=0}^{\infty}\sum_{j=0}^{\infty} (-1)^{k+j} e^{j+2} \binom{c_t-2}{k}\binom{a(k+2)-2}{j} I_{\text{comp}}(j)\]
\end{proof}

\begin{lemma}
	Let $Y_t$ be a random variable whose conditional distribution given
	$\mathcal{F}_{t-1}$ is specified by KT$(\mu_{\rho,t}, a, \lambda)$. Then
	\[\mathbb{E}\left(\frac{y_t e^{\psi(y_t;\lambda)} (e^{\lambda y_t} -1) }{A_t} \;\middle|\; \mathcal{F}_{t-1}\right) = \frac{a c_t}{\lambda} \sum_{k=0}^{\infty} \sum_{j=0}^{\infty} (-1)^{k+j} e^{j+2} \binom{c_t - 1}{k} \binom{a(k+1)-2}{j} I_{\text{comp}}(j) \]
	\label{lem:lemma5}
\end{lemma}
\begin{proof}
	\begin{align*}
		\mathbb{E}\left(\frac{y_t e^{\psi(y_t;\lambda)} (e^{\lambda y_t} -1) }{A_t} \;\middle|\; \mathcal{F}_{t-1}\right) 
		&= 	a \lambda c_t \int_0^\infty y (e^{\lambda y} - 1)^2 e^{2\psi(y;\lambda)} A_t^{a-2} (1 - A_t^a)^{c_t-1} \, dy
	\end{align*}
	By utilizing the binomial expansion and variable transformations established in Lemma \ref{lem:lemma4}, and subsequently evaluating the resulting integral using properties of the upper incomplete gamma function in an identical fashion, we obtain:
	\[\mathbb{E}\left(\frac{y_t e^{\psi(y_t;\lambda)} (e^{\lambda y_t} -1) }{A_t} \;\middle|\; \mathcal{F}_{t-1}\right) = \frac{a c_t}{\lambda} \sum_{k=0}^{\infty} \sum_{j=0}^{\infty} (-1)^{k+j} e^{j+2} \binom{c_t - 1}{k} \binom{a(k+1)-2}{j} I_{\text{comp}}(j)\]
\end{proof}
\begin{lemma}
	Let $Y_t$ be a random variable whose conditional distribution given
	$\mathcal{F}_{t-1}$ is specified by KT$(\mu_{\rho,t}, a, \lambda)$. Then
	\[\mathbb{E}\left(\frac{A_t^{a-1} y_t e^{\psi(y_t;\lambda)} (e^{\lambda y_t} -1) \log(A_t)}{(1-A_t^a)^2} \;\middle|\; \mathcal{F}_{t-1}\right) = -\frac{a c_t}{\lambda} \sum_{k=0}^{\infty} \sum_{j=0}^{\infty} \sum_{m=1}^{\infty} \frac{(-1)^{k+j}}{m} e^{L} \binom{c_t - 3}{k} \binom{a(k+2)-2}{j} I_{\text{comp}}(L) \]
	\label{lem:lemma6}
\end{lemma}

\begin{proof}
	\begin{align*}
		&\mathbb{E}\left(\frac{A_t^{a-1} y_t e^{\psi(y;\lambda)} (e^{\lambda y_t} -1) \log(A_t) }{(1-A_t^a)^2} \;\middle\vert{}\; \mathcal{F}_{t-1}\right)
		&= 	a \lambda c_t \int_0^\infty y (e^{\lambda y} - 1)^2 e^{2\psi(y;\lambda)} \log(A_t) A_t^{2a-2} (1 - A_t^a)^{c_t-3} \, dy
	\end{align*}
	By using binomial expansion and transformation done in Lemma \ref{lem:lemma4} along with expanding Taylor series \(\log(A_t) = \log(1 - e^{\psi(y;\lambda)}) = -\sum_{m=1}^{\infty} \frac{1}{m} e^{m\psi(y;\lambda)}\), we have:
\[	-\frac{a c_t e^L}{\lambda} \sum_{k=0}^{\infty} \sum_{j=0}^{\infty} \sum_{m=1}^{\infty} \frac{(-1)^{k+j}}{m} \binom{c_t - 3}{k} \binom{a(k+2)-2}{j} \int_1^\infty \ln(z) \left[ z^{L+1} - 2z^L + z^{L-1} \right] e^{-L z} \, dz \]
where, $L = j + m + 2$.
Now, by using the properties of the upper incomplete gamma function as:
		\begin{align*}
			I_{\text{comp}}(L) = &\frac{1}{L^{L+2}} \left[ \Gamma'(L+2, L) - \ln(L) \Gamma(L+2, L) \right] - \frac{2}{L^{L+1}} \left[ \Gamma'(L+1, L) - \ln(L) \Gamma(L+1, L) \right] \\
			&+ \frac{1}{L^L} \left[ \Gamma'(L, L) - \ln(L) \Gamma(L, L) \right]
		\end{align*}
Therefore, the final solution is given by:
\[\mathbb{E}\left(\frac{A_t^{a-1} y_t e^{\psi(y_t;\lambda)} (e^{\lambda y_t} -1) \log(A_t)}{(1-A_t^a)} \;\middle|\; \mathcal{F}_{t-1}\right) = -\frac{a c_t}{\lambda} \sum_{k=0}^{\infty} \sum_{j=0}^{\infty} \sum_{m=1}^{\infty} \frac{(-1)^{k+j}}{m} e^{L} \binom{c_t - 3}{k} \binom{a(k+2)-2}{j} I_{\text{comp}}(L)\]
\end{proof}

\begin{lemma}
	Let $Y_t$ be a random variable whose conditional distribution given $\mathcal{F}_{t-1}$ is specified by KT$(\mu_{\rho,t}, a, \lambda)$. Then
\[\mathbb{E}\left(\frac{ y_t^2 e^{\psi(y_t;\lambda)} e^{\lambda y_t} }{A_t} \;\middle|\; \mathcal{F}_{t-1}\right) = \frac{a c_t}{\lambda^2} \sum_{k=0}^{\infty} \sum_{j=0}^{\infty} (-1)^{k+j} e^{j+2} \binom{c_t - 1}{k} \binom{a(k+1)-2}{j} M_{\text{comp}}(j) \]
	\label{lem:lemma7}
\end{lemma}
\begin{proof}
\[	\mathbb{E}\left(\frac{ y_t^2 e^{\psi(y_t;\lambda)} e^{\lambda y_t} }{A_t} \;\middle|\; \mathcal{F}_{t-1}\right) = 	a \lambda c_t \int_0^\infty y^2 \left[ e^{2\lambda y} - e^{\lambda y} \right] e^{2\psi(y;\lambda)} A_t^{a-2} (1 - A_t^a)^{c_t-1} \, dy\]
	Applying the binomial series and power expansions sequentially as established in the preceding derivations simplifies the integrand into a dual summation:
\[	\frac{a c_t e^{j+2}}{\lambda^2} \sum_{k=0}^{\infty} \sum_{j=0}^{\infty} (-1)^{k+j} \binom{c_t - 1}{k} \binom{a(k+1)-2}{j} \int_1^\infty (\ln(z))^2 \left[ z^{j+3} - z^{j+2} \right] e^{-(j+2)z} \, dz \]
Furthermore, by introducing the second-order derivative of the upper incomplete gamma function via the tracking function $M_{\text{comp}}(j)$, it follows that:
	\begin{align*}
		M_{\text{comp}}(j) = &\frac{1}{b^{j+4}} \left[ \Gamma''(j+4, b) - 2\ln(b)\Gamma'(j+4, b) + (\ln(b))^2\Gamma(j+4, b) \right] \\
		&- \frac{1}{b^{j+3}} \left[ \Gamma''(j+3, b) - 2\ln(b)\Gamma'(j+3, b) + (\ln(b))^2\Gamma(j+3, b) \right]
	\end{align*}
	where $b = j+2$, $\Gamma(m, b) = \int_b^\infty t^{m-1}e^{-t}dt$, $\Gamma'(m, b) = \frac{\partial}{\partial m}\Gamma(m, b)$, and $\Gamma''(m, b) = \frac{\partial^2}{\partial m^2}\Gamma(m, b)$. Consequently, the closed-form analytical solution is expressed as:
	\[	\mathbb{E}\left(\frac{ y_t^2 e^{\psi(y_t;\lambda)} e^{\lambda y_t} }{A_t} \;\middle|\; \mathcal{F}_{t-1}\right) = \frac{a c_t}{\lambda^2} \sum_{k=0}^{\infty} \sum_{j=0}^{\infty} (-1)^{k+j} e^{j+2} \binom{c_t - 1}{k} \binom{a(k+1)-2}{j} M_{\text{comp}}(j) \]
\end{proof}

\begin{lemma}
	Let $Y_t$ be a random variable whose conditional distribution given $\mathcal{F}_{t-1}$ is specified by KT$(\mu_{\rho,t}, a, \lambda)$. Then
	\[ \mathbb{E}\left(\frac{ y_t^2 e^{\psi(y_t;\lambda)} (e^{\lambda y_t} -1)^2  }{A_t} \;\middle|\; \mathcal{F}_{t-1}\right) = \frac{a c_t}{\lambda^2} \sum_{k=0}^{\infty} \sum_{j=0}^{\infty} (-1)^{k+j} e^{j+2} \binom{c_t - 1}{k} \binom{a(k+1)-2}{j} N_{\text{comp}}(j) \]
	\label{lem:lemma8}
\end{lemma}
\begin{proof}
	\[ \mathbb{E}\left(\frac{ y_t^2 e^{\psi(y_t;\lambda)} (e^{\lambda y_t} -1)^2  }{A_t} \;\middle|\; \mathcal{F}_{t-1}\right) = a \lambda c_t \int_0^\infty y^2 \left[ e^{3\lambda y} - 2e^{2\lambda y} + e^{\lambda y} \right] e^{2\psi(y;\lambda)} A_t^{a-2} (1 - A_t^a)^{c_t-1} \, dy \]
	Following the binomial and exponential expansion detailed in the prior derivations, the integral becomes:
	\[ \frac{a c_t e^{j+2}}{\lambda^2} \sum_{k=0}^{\infty} \sum_{j=0}^{\infty} (-1)^{k+j} \binom{c_t - 1}{k} \binom{a(k+1)-2}{j} \int_1^\infty (\ln(z))^2 \left[ z^{j+4} - 2z^{j+3} + z^{j+2} \right] e^{-(j+2)z} \, dz \]
	Furthermore, by using second-order derivative of the upper incomplete gamma, we have:
	\begin{align*}
		N_{\text{comp}}(j) = &\frac{1}{b^{j+5}} \left[ \Gamma''(j+5, b) - 2\ln(b)\Gamma'(j+5, b) + (\ln(b))^2\Gamma(j+5, b) \right] \\
		&- \frac{2}{b^{j+4}} \left[ \Gamma''(j+4, b) - 2\ln(b)\Gamma'(j+4, b) + (\ln(b))^2\Gamma(j+4, b) \right] \\
		&+ \frac{1}{b^{j+3}} \left[ \Gamma''(j+3, b) - 2\ln(b)\Gamma'(j+3, b) + (\ln(b))^2\Gamma(j+3, b) \right]
	\end{align*}
	where $b = j+2$, $\Gamma(m, b) = \int_b^\infty t^{m-1}e^{-t}dt$, $\Gamma'(m, b) = \frac{\partial}{\partial m}\Gamma(m, b)$, and $\Gamma''(m, b) = \frac{\partial^2}{\partial m^2}\Gamma(m, b)$. 
	Hence, the expectation is expressed as:
	\[ \mathbb{E}\left(\frac{ y_t^2 e^{\psi(y_t;\lambda)} (e^{\lambda y_t} -1)^2  }{A_t} \;\middle|\; \mathcal{F}_{t-1}\right) = \frac{a c_t}{\lambda^2} \sum_{k=0}^{\infty} \sum_{j=0}^{\infty} (-1)^{k+j} e^{j+2} \binom{c_t - 1}{k} \binom{a(k+1)-2}{j} N_{\text{comp}}(j) \]
\end{proof}

\begin{lemma}
	Let $Y_t$ be a random variable whose conditional distribution given $\mathcal{F}_{t-1}$ is specified by KT$(\mu_{\rho,t}, a, \lambda)$. Then
	\[ \mathbb{E}\left(\frac{ y_t^2 e^{2\psi(y_t;\lambda)} (e^{\lambda y_t} -1)^2  }{A_t^2} \;\middle|\; \mathcal{F}_{t-1}\right) = \frac{a c_t}{\lambda^2} \sum_{k=0}^{\infty} \sum_{j=0}^{\infty} (-1)^{k+j} e^{j+3} \binom{c_t - 1}{k} \binom{a(k+1)-3}{j} P_{\text{comp}}(j) \]
	\label{lem:lemma9}
\end{lemma}
\begin{proof}
	\[ \mathbb{E}\left(\frac{ y_t^2 e^{2\psi(y_t;\lambda)} (e^{\lambda y_t} -1)^2  }{A_t^2} \;\middle|\; \mathcal{F}_{t-1}\right) = a \lambda c_t \int_0^\infty y^2 \left[ e^{3\lambda y} - 2e^{2\lambda y} + e^{\lambda y} \right] e^{3\psi(y;\lambda)} A_t^{a-3} (1 - A_t^a)^{c_t-1} \, dy \]
	By executing the same binomial and power expansion introduced in the previous lemmas, the integral simplifies to:
	\[ \frac{a c_t e^{j+3}}{\lambda^2} \sum_{k=0}^{\infty} \sum_{j=0}^{\infty} (-1)^{k+j} \binom{c_t - 1}{k} \binom{a(k+1)-3}{j} \int_1^\infty (\ln(z))^2 \left[ z^{j+5} - 2z^{j+4} + z^{j+3} \right] e^{-(j+3)z} \, dz \]
	Now, by taking the second-order derivative of the upper incomplete gamma function, we get:
	\begin{align*}
		P_{\text{comp}}(j) = &\frac{1}{b^{j+6}} \left[ \Gamma''(j+6, b) - 2\ln(b)\Gamma'(j+6, b) + (\ln(b))^2\Gamma(j+6, b) \right] \\
		&- \frac{2}{b^{j+5}} \left[ \Gamma''(j+5, b) - 2\ln(b)\Gamma'(j+5, b) + (\ln(b))^2\Gamma(j+5, b) \right] \\
		&+ \frac{1}{b^{j+4}} \left[ \Gamma''(j+4, b) - 2\ln(b)\Gamma'(j+4, b) + (\ln(b))^2\Gamma(j+4, b) \right]
	\end{align*}
	where $b = j+3$, $\Gamma(m, b) = \int_b^\infty t^{m-1}e^{-t}dt$, $\Gamma'(m, b) = \frac{\partial}{\partial m}\Gamma(m, b)$, and $\Gamma''(m, b) = \frac{\partial^2}{\partial m^2}\Gamma(m, b)$. 
	Therefore, the final solution is expressed as:
	\[ \mathbb{E}\left(\frac{ y_t^2 e^{2\psi(y_t;\lambda)} (e^{\lambda y_t} -1)^2  }{A_t^2} \;\middle|\; \mathcal{F}_{t-1}\right) = \frac{a c_t}{\lambda^2} \sum_{k=0}^{\infty} \sum_{j=0}^{\infty} (-1)^{k+j} e^{j+3} \binom{c_t - 1}{k} \binom{a(k+1)-3}{j} P_{\text{comp}}(j) \]
\end{proof}

\begin{lemma}
	Let $Y_t$ be a random variable whose conditional distribution given $\mathcal{F}_{t-1}$ is specified by KT$(\mu_{\rho,t}, a, \lambda)$. Then
	\[ \mathbb{E}\left(\frac{A_t^{a-1} y_t^2 e^{\lambda y_t}  e^{\psi(y_t;\lambda)} }{(1-A_t^a)} \;\middle|\; \mathcal{F}_{t-1}\right) = \frac{a c_t}{\lambda^2} \sum_{k=0}^{\infty} \sum_{j=0}^{\infty} (-1)^{k+j} e^{j+2} \binom{c_t - 2}{k} \binom{a(k+2)-2}{j} M_{\text{comp}}(j) \]
	\label{lem:lemma10}
\end{lemma}
\begin{proof}
	\[ \mathbb{E}\left(\frac{A_t^{a-1} y_t^2 e^{\lambda y_t}  e^{\psi(y_t;\lambda)} }{(1-A_t^a)} \;\middle|\; \mathcal{F}_{t-1}\right) = a \lambda c_t \int_0^\infty y^2 \left[ e^{2\lambda y} - e^{\lambda y} \right] e^{2\psi(y;\lambda)} A_t^{2a-2} (1 - A_t^a)^{c_t-2} \, dy \]
	Expanding the algebraic terms binomially and aggregating powers over the baseline distribution fields leaves the following dual series structure:
	\[ \frac{a c_t e^{j+2}}{\lambda^2} \sum_{k=0}^{\infty} \sum_{j=0}^{\infty} (-1)^{k+j} \binom{c_t - 2}{k} \binom{a(k+2)-2}{j} \int_1^\infty (\ln(z))^2 \left[ z^{j+3} - z^{j+2} \right] e^{-(j+2)z} \, dz \]
	Since, the integration term is similar to Lemma \ref{lem:lemma7}, the expectation is given by:
	\[ \mathbb{E}\left(\frac{A_t^{a-1} y_t^2 e^{\lambda y_t}  e^{\psi(y_t;\lambda)} }{(1-A_t^a)} \;\middle|\; \mathcal{F}_{t-1}\right) = \frac{a c_t}{\lambda^2} \sum_{k=0}^{\infty} \sum_{j=0}^{\infty} (-1)^{k+j} e^{j+2} \binom{c_t - 2}{k} \binom{a(k+2)-2}{j} M_{\text{comp}}(j) \]
\end{proof}

\begin{lemma}
	Let $Y_t$ be a random variable whose conditional distribution given $\mathcal{F}_{t-1}$ is specified by KT$(\mu_{\rho,t}, a, \lambda)$. Then
	\[ \mathbb{E}\left(\frac{A_t^{a-1} y_t^2 (e^{\lambda y_t} -1)^2 e^{\psi(y_t;\lambda)} }{(1-A_t^a)} \;\middle|\; \mathcal{F}_{t-1}\right) = \frac{a c_t}{\lambda^2} \sum_{k=0}^{\infty} \sum_{j=0}^{\infty} (-1)^{k+j} e^{j+2} \binom{c_t - 2}{k} \binom{a(k+2)-2}{j} N_{\text{comp}}(j) \]
	\label{lem:lemma11}
\end{lemma}
\begin{proof}
	\[ \mathbb{E}\left(\frac{A_t^{a-1} y_t^2 (e^{\lambda y_t} -1)^2 e^{\psi(y_t;\lambda)} }{(1-A_t^a)} \;\middle|\; \mathcal{F}_{t-1}\right) = a \lambda c_t \int_0^\infty y^2 \left[ e^{3\lambda y} - 2e^{2\lambda y} + e^{\lambda y} \right] e^{2\psi(y;\lambda)} A_t^{2a-2} (1 - A_t^a)^{c_t-2} \, dy \]
	Applying the binomial series and power expansions similar to the preceding derivations follows:
	\[ \frac{a c_t e^{j+2}}{\lambda^2} \sum_{k=0}^{\infty} \sum_{j=0}^{\infty} (-1)^{k+j} \binom{c_t - 2}{k} \binom{a(k+2)-2}{j} \int_1^\infty (\ln(z))^2 \left[ z^{j+4} - 2z^{j+3} + z^{j+2} \right] e^{-(j+2)z} \, dz \]
	Following the identical integration detailed in Lemma \ref{lem:lemma8}, the complete analytical expectation simplifies to:
	\[ \mathbb{E}\left(\frac{A_t^{a-1} y_t^2 (e^{\lambda y_t} -1)^2 e^{\psi(y_t;\lambda)} }{(1-A_t^a)} \;\middle|\; \mathcal{F}_{t-1}\right) = \frac{a c_t}{\lambda^2} \sum_{k=0}^{\infty} \sum_{j=0}^{\infty} (-1)^{k+j} e^{j+2} \binom{c_t - 2}{k} \binom{a(k+2)-2}{j} N_{\text{comp}}(j) \]
\end{proof}

\begin{lemma}
	Let $Y_t$ be a random variable whose conditional distribution given $\mathcal{F}_{t-1}$ is specified by KT$(\mu_{\rho,t}, a, \lambda)$. Then
	\[ \mathbb{E}\left(\frac{A_t^{a-2} y_t^2 (e^{\lambda y_t} -1)^2 e^{2 \psi(y_t;\lambda)} }{(1-A_t^a)} \;\middle|\; \mathcal{F}_{t-1}\right) = \frac{a c_t}{\lambda^2} \sum_{k=0}^{\infty} \sum_{j=0}^{\infty} (-1)^{k+j} e^{j+3} \binom{c_t - 2}{k} \binom{a(k+2)-3}{j} P_{\text{comp}}(j) \]
	\label{lem:lemma12}
\end{lemma}
\begin{proof}
	\[ \mathbb{E}\left(\frac{A_t^{a-2} y_t^2 (e^{\lambda y_t} -1)^2 e^{2 \psi(y_t;\lambda)} }{(1-A_t^a)} \;\middle|\; \mathcal{F}_{t-1}\right) = a \lambda c_t \int_0^\infty y^2 \left[ e^{3\lambda y} - 2e^{2\lambda y} + e^{\lambda y} \right] e^{3\psi(y;\lambda)} A_t^{2a-3} (1 - A_t^a)^{c_t-2} \, dy \]
 Decomposing the fractional power terms via generalized binomial series leads directly to the following expression:
	\[ \frac{a c_t e^{j+3}}{\lambda^2} \sum_{k=0}^{\infty} \sum_{j=0}^{\infty} (-1)^{k+j} \binom{c_t - 2}{k} \binom{a(k+2)-3}{j} \int_1^\infty (\ln(z))^2 \left[ z^{j+5} - 2z^{j+4} + z^{j+3} \right] e^{-(j+3)z} \, dz \]
Recognizing that the integral structurally mirrors the configuration in Lemma \ref{lem:lemma9}, the final exact evaluation simplifies to:
	\[ \mathbb{E}\left(\frac{A_t^{a-2} y_t^2 (e^{\lambda y_t} -1)^2 e^{2 \psi(y_t;\lambda)} }{(1-A_t^a)} \;\middle|\; \mathcal{F}_{t-1}\right) = \frac{a c_t}{\lambda^2} \sum_{k=0}^{\infty} \sum_{j=0}^{\infty} (-1)^{k+j} e^{j+3} \binom{c_t - 2}{k} \binom{a(k+2)-3}{j} P_{\text{comp}}(j) \]
\end{proof}

\begin{lemma}
	Let $Y_t$ be a random variable whose conditional distribution given $\mathcal{F}_{t-1}$ is specified by KT$(\mu_{\rho,t}, a, \lambda)$. Then
	\[ \mathbb{E}\left(\frac{A_t^{a-2} y_t^2 (e^{\lambda y_t} -1)^2 e^{2 \psi(y_t;\lambda)} }{(1-A_t^a)^2} \;\middle|\; \mathcal{F}_{t-1}\right) = \frac{a c_t}{\lambda^2} \sum_{k=0}^{\infty} \sum_{j=0}^{\infty} (-1)^{k+j} e^{j+3} \binom{c_t - 3}{k} \binom{a(k+2)-3}{j} P_{\text{comp}}(j) \]
	\label{lem:lemma13}
\end{lemma}
\begin{proof}
	\[ \mathbb{E}\left(\frac{A_t^{a-2} y_t^2 (e^{\lambda y_t} -1)^2 e^{2 \psi(y_t;\lambda)} }{(1-A_t^a)^2} \;\middle|\; \mathcal{F}_{t-1}\right) = a \lambda c_t \int_0^\infty y^2 \left[ e^{3\lambda y} - 2e^{2\lambda y} + e^{\lambda y} \right] e^{3\psi(y;\lambda)} A_t^{2a-3} (1 - A_t^a)^{c_t-3} \, dy \]
	Deploying the generalized binomial theorem and power expansions sequentially under the framework established in the preceding proofs reduces the integrand to a dual summation:
	\[ \frac{a c_t e^{j+3}}{\lambda^2} \sum_{k=0}^{\infty} \sum_{j=0}^{\infty} (-1)^{k+j} \binom{c_t - 3}{k} \binom{a(k+2)-3}{j} \int_1^\infty (\ln(z))^2 \left[ z^{j+5} - 2z^{j+4} + z^{j+3} \right] e^{-(j+3)z} \, dz \]
By leveraging the mathematical symmetry with the integral resolved in Lemma \ref{lem:lemma9}, the complete analytical expectation is given by:
	\[ \mathbb{E}\left(\frac{A_t^{a-2} y_t^2 (e^{\lambda y_t} -1)^2 e^{2 \psi(y_t;\lambda)} }{(1-A_t^a)^2} \;\middle|\; \mathcal{F}_{t-1}\right) = \frac{a c_t}{\lambda^2} \sum_{k=0}^{\infty} \sum_{j=0}^{\infty} (-1)^{k+j} e^{j+3} \binom{c_t - 3}{k} \binom{a(k+2)-3}{j} P_{\text{comp}}(j) \]
\end{proof}

\begin{lemma}
	Let $Y_t$ be a random variable whose conditional distribution given $\mathcal{F}_{t-1}$ is specified by KT$(\mu_{\rho,t}, a, \lambda)$. Then
	\[ \mathbb{E}\left(y_t^2 e^{\lambda y_t} \;\middle|\; \mathcal{F}_{t-1}\right) = \frac{a c_t}{\lambda^2} \sum_{k=0}^{\infty} \sum_{j=0}^{\infty} (-1)^{k+j} e^{j+1} \binom{c_t - 1}{k} \binom{a(k+1)-1}{j} L_{\text{comp}}(j) \]
	\label{lem:lemma14}
\end{lemma}
\begin{proof}
	\[ \mathbb{E}\left(y_t^2 e^{\lambda y_t} \;\middle|\; \mathcal{F}_{t-1}\right) = a \lambda c_t \int_0^\infty y^2 \left[ e^{2\lambda y} - e^{\lambda y} \right] e^{\psi(y;\lambda)} A_t^{a-1} (1 - A_t^a)^{c_t-1} \, dy \]
	By using the binomial expansion and transformation as done previously, we get:
	\[ \frac{a c_t e^{j+1}}{\lambda^2} \sum_{k=0}^{\infty} \sum_{j=0}^{\infty} (-1)^{k+j} \binom{c_t - 1}{k} \binom{a(k+1)-1}{j} \int_1^\infty (\ln(z))^2 \left[ z^{j+2} - z^{j+1} \right] e^{-(j+1)z} \, dz \]
	Furthermore, by introducing the second-order derivative configuration of the upper incomplete gamma function via the tracking function $L_{\text{comp}}(j)$, it follows that:
	\begin{align*}
		L_{\text{comp}}(j) = &\frac{1}{b^{j+3}} \left[ \Gamma''(j+3, b) - 2\ln(b)\Gamma'(j+3, b) + (\ln(b))^2\Gamma(j+3, b) \right] \\
		&- \frac{1}{b^{j+2}} \left[ \Gamma''(j+2, b) - 2\ln(b)\Gamma'(j+2, b) + (\ln(b))^2\Gamma(j+2, b) \right]
	\end{align*}
	where $b = j+1$, $\Gamma(m, b) = \int_b^\infty t^{m-1}e^{-t}dt$, $\Gamma'(m, b) = \frac{\partial}{\partial m}\Gamma(m, b)$, and $\Gamma''(m, b) = \frac{\partial^2}{\partial m^2}\Gamma(m, b)$. Consequently, the closed-form analytical solution is expressed as:
	\[ \mathbb{E}\left(y_t^2 e^{\lambda y_t} \;\middle|\; \mathcal{F}_{t-1}\right) = \frac{a c_t}{\lambda^2} \sum_{k=0}^{\infty} \sum_{j=0}^{\infty} (-1)^{k+j} e^{j+1} \binom{c_t - 1}{k} \binom{a(k+1)-1}{j} L_{\text{comp}}(j) \]
\end{proof}

\begin{lemma}
	Let $Y_t$ be a random variable whose conditional distribution given $\mathcal{F}_{t-1}$ is specified by KT$(\mu_{\rho,t}, a, \lambda)$. Then the structural expectation of the baseline density quotient is formulated as:
	\[ \mathbb{E}\left(\frac{y_t^2 e^{\lambda y_t}}{(e^{\lambda y_t} -1)^2} \;\middle|\; \mathcal{F}_{t-1}\right) = \frac{a c_t}{\lambda^2} \sum_{k=0}^{\infty} \sum_{j=0}^{\infty} (-1)^{k+j} e^{j+1} \binom{c_t - 1}{k} \binom{a(k+1)-1}{j} W_{\text{comp}}(j) \]
	\label{lem:lemma15}
\end{lemma}
\begin{proof}
	\[ \mathbb{E}\left(\frac{y_t^2 e^{\lambda y_t}}{(e^{\lambda y_t} -1)^2} \;\middle|\; \mathcal{F}_{t-1}\right) = a \lambda c_t \int_0^\infty \frac{y^2 e^{\lambda y}}{e^{\lambda y} - 1} e^{\psi(y;\lambda)} A_t^{a-1} (1 - A_t^a)^{c_t-1} \, dy \]
	Applying the binomial expansion and power transformation sequentially for the geometric components as derived above simplifies the integrand into a dual summation:
	\[ \frac{a c_t e^{j+1}}{\lambda^2} \sum_{k=0}^{\infty} \sum_{j=0}^{\infty} (-1)^{k+j} \binom{c_t - 1}{k} \binom{a(k+1)-1}{j} \int_1^\infty (\ln(z))^2 \left[ \frac{z^{j+1}}{z - 1} \right] e^{-(j+1)z} \, dz \]
	By using the second-order derivative configuration of the upper incomplete gamma function, we can define $W_{\text{comp}}(j)$ explicitly as:
	\[ W_{\text{comp}}(j) = \sum_{m=0}^{\infty} \frac{1}{b^{j-m+1}} \left[ \Gamma''(j-m+1, b) - 2\ln(b)\Gamma'(j-m+1, b) + (\ln(b))^2\Gamma(j-m+1, b) \right] \]
	Therefore, the expression is given by:
	\[ \mathbb{E}\left(\frac{y_t^2 e^{\lambda y_t}}{(e^{\lambda y_t} -1)^2} \;\middle|\; \mathcal{F}_{t-1}\right) = \frac{a c_t}{\lambda^2} \sum_{k=0}^{\infty} \sum_{j=0}^{\infty} (-1)^{k+j} e^{j+1} \binom{c_t - 1}{k} \binom{a(k+1)-1}{j} W_{\text{comp}}(j) \]
\end{proof}

\bibliography{test}

\end{document}